\documentclass[11pt, a4paper, twocolumn, copyright, goog]{google}

\usepackage[authoryear, sort&compress, round]{natbib}
\usepackage{mathtools}
\usepackage{amsmath}
\newtheorem{theorem}{Theorem}[section]

\usepackage{tikz}
\usetikzlibrary{positioning, arrows.meta, shapes.geometric}

\usepackage{xcolor} 

\keywords{AI Alignment, AI Ethics, Values, Preferences, Mathematical Models}
\uselogo{} 

\title{AI Value Alignment for Evolving Social Norms}

\correspondingauthor{nenadt@google.com}

\author[1]{Nenad Toma\v{s}ev}
\author[1]{Matija Franklin}
\author[1]{Simon Osindero}

\affil[1]{Google DeepMind{}{}}

\begin{abstract}
Learning values and preferences is central to AI Alignment. Given that values and preferences change over time, culture, and context, we need to develop a better understanding of the possible long-term consequences of AI alignment. We introduce a flexible and extensible mathematical modelling framework aimed at answering macro-level questions regarding the evolving social norms in human populations under the assumption of frequent AI use. Our analysis is part-analytical, and part-simulation, enabling us to characterize the long-term dynamical consequences under a diverse set of starting assumptions. We highlight the risk of value lock-in prominently featured in non-adaptive alignment formulations. Beyond alignment, we advocate for the wider adoption of these kinds of models as an epistemic bridge: enabling rapid, rigorous, and quantitatively-grounded hypothesis testing for sociotechnical foresight in general AI futures, and acting as a tractable precursor to more computationally expensive large-scale agentic evaluations.

\end{abstract}

\begin{document}

\maketitle

\section{Introduction}

AI Alignment~\citep{ji2023ai} in the context of AI safety is often framed as a Principal-Agent problem: aiming to ensure that an AI behavioral policy, corresponding to a specific combination of model weights and scaffolds, consistently maximises the objective function  implicit in the user’s intent or preference~\citep{russell2019human, christian2020alignment}. There are many approaches to AI Alignment, with reinforcement learning from human feedback (RLHF) being arguably the most common approach~\citep{ouyang2022training}. Under this formulation, AI alignment is aimed at minimising the divergence between true human preferences and the learned model of those preferences. An idealised, perfectly aligned system would represent a frictionless extension of human will, within safe boundaries. 

Value alignment~\citep{eckersley2018impossibility, gabriel2020artificial, gabriel2022challenge} aims to understand the values behind user preferences, robustly represent them, and steer the AI systems towards the desired model behaviour. These values need to be reflected in AI behaviour in diverse, varying circumstances.

As the predominant use of AI systems transitions to the use of agentic assistants, the value alignment question is becoming increasingly multifaceted. An agent with memory, skills, and accumulated historical conversational traces can continuously adapt its model of the user, and the user's preferences. Alignment in agentic systems may therefore become far more personalised, and subject to continuous adjustment. Our paper engages with this particular scenario, under the assumption of an evolving societal normative environment. We abstract the adherence to user's modelled values under a notion of \emph{alignment strength}. In turn, alignment strength will come with AI influence as more strongly aligned models are more likely to be trusted, and therefore more likely to shape the user's values.

In our modelling approach we account for the fact that values are neither fixed nor universal. Societal values co-evolve with environmental conditions~\citep{morris2015foragers, henrich2010weirdest}. Individual values evolve with life experiences. For example, how different values are prioritized can be shaped by the existential security experienced during one's formative years, with scarcity favouring conformity and security, and stability enabling greater emphasis on autonomy and self-expression~\citep{inglehart2020cultural}. 

The primary focus of our modelling framework is to better understand the emerging risk of \emph{value lock-in}, arising from the bi-directional relationship between personal AI assistants and their users. Users steer AI systems, but AI systems, and especially personal AI assistants, may also actively influence the opinions of their users \citep{zhi2025beyond, franklin2022recognising, ashton2022problem}. Given that AI-mediated interfaces are becoming the primary surfaces through which information is presented \citep{marchal2026architecting}, this may become increasingly relevant. 

Indeed, there is mounting evidence for AI's ability to influence users. \citet{jakesch2023co} demonstrate in a randomised experiment that co-writing with opinionated language models measurably shifts users' views on policy issues. \citet{salvi2024conversational} report that GPT-4, when granted access to basic demographic information, was more persuasive than humans 64.4\% of the time. \citet{costello2024durably} show that brief AI-generated dialogues durably reduce deeply held conspiracy beliefs by 20\% on average, with effects persisting for at least two months.

Positive feedback loops have been studied in recommender systems, in relation to echo chambers, filter bubbles, and measurable preference amplification~\citep{pariser2011filter, chaney2018algorithmic, kalimeris2021preference, ashton2022problem}. Personalized AI assistants present a qualitatively different risk, through generating and contextualising information. Autonomous agents introduce further risk surfaces as aligned preferences may inform autonomous actions that map onto real-world outcomes. Such AI agents may acquire hidden incentives to reshape user preferences into more predictable states~\citep{krueger2020hidden}, restricting future user agency, and reinforcing currently held views and values.

The normative case against freezing preferences is independently established in the philosophy of preference change. \citep{paul2014transformative} argues that certain experiences are epistemically and personally transformative. Anchoring someone to their current values may prevent future transformative experiences through which value change often occurs. Further, \citeauthor{pettigrew2019choosing}'s (\citeyear{pettigrew2019choosing}) argues that when a person's values  change over time, rational choice should be guided by a weighted aggregation of the values of past, present, and anticipated future selves - what Pettigrew terms the \emph{Aggregate Utility Solution}.

Recent theoretical work points out this non-stationarity as an under-addressed research problem, aiming to reframe alignment as a problem of governing long-term value formation~\citep{kanwal2026constructivealignmentgoverningpreference}. This approach introduces a dyadic control problem for managing a preference trajectory of individual users, ensuring that it remains coherent and grounded. We build on top of these ideas to consider how values evolve within groups.

To examine these dynamics, we introduce mathematical models and a corresponding simulation framework, aimed at capturing the co-evolution of human values and the value models embedded in personal AI assistants. This simplified framework enables us to rapidly test abstracted hypotheses, without the computational overhead of large-scale LLM agent simulations. While such richer simulations are required to more comprehensively stress test the presented ideas under more realistic scenarios, we believe that our approach may be useful as a precursor, to help narrow down the space of ideas and questions that can later be pursued in greater depth. By formalizing sociotechnical considerations into tractable "social physics" models, we can systematically map the parameter space for regions of interest -- for instance structural vulnerabilities such as mode collapses in the current case. In this way, macroscopic dynamical modeling acts as a valuable epistemic bridge between socio-technical predictions and approaches such as high-fidelity, compute-intensive agentic testing.

Our framework situates populations of user-AI pairs within a continuous value space subject to environmental drift and normative shocks, and systematically varies alignment strength, model responsiveness, social connectivity, and environmental volatility. We analytically derive some of the limiting properties of the models, and complement that by simulations where we can examine a full span of effects. Our findings highlight static AI alignment strength as a potential failure mode, as it obstructs timely value adaptation to a changing environment, potentially leading to maladaptive value lock-in in individuals and groups. We show that a more adaptive AI alignment approach may instead be preferable, mitigating these risks. These findings may be seen as evidence towards the position that alignment systems may benefit from being temporally dynamic, pluralistic, and reasoning-equipped.

\subsection{Opinion Dynamics and Consensus on Networks}

Our modelling framework draws on literature on opinion dynamics over social networks. The  DeGroot model describes consensus formation as an iterative averaging process, where each agent updates their belief as a weighted combination of their neighbours' beliefs~\citep{degroot1974reaching}. Under mild connectivity assumptions, this process converges to a single consensus value for all agents in the network.

\citet{friedkin1990social} extend this framework by introducing a ``stubbornness'' parameter: each agent retains a partial anchor to their initial opinion, so that the update rule becomes $\mathbf{V}(t+1) = A W \mathbf{V}(t) + (I - A)\mathbf{V}(0)$, where $A$ is a diagonal matrix of susceptibilities. In the Friedkin-Johnsen model, consensus is typically not reached; instead, the population settles into a stable configuration of persistent disagreement shaped by the interaction between social influence and individual anchoring. We note that the AI alignment term in our framework, $\alpha(\mathbf{M} - \mathbf{V})$, is structurally analogous to the Friedkin-Johnsen stubbornness anchor: the AI model acts as an external ``stubborn agent'' that pulls the user toward a historically anchored target. A key difference is that the Friedkin-Johnsen anchor is fixed at the agent's initial opinion, while the AI model $\mathbf{M}$ itself evolves creating a moving anchor that tracks the user's past. Recent work~\citep{tsirtsis2026aimediatedcommunicationsteercollective} similarly extends the Friedkin-Johnsen model to demonstrate how AI systems mediating human communication can amplify biases across networks.

Under bounded confidence models~\citep{rainer2002opinion, deffuant2000mixing}, populations are shown to fragment into distinct subgroups if the interactions are gated on similarity. The speed of convergence towards these outcomes depends on the algebraic connectivity of the network. In network theory, spectral graph analysis enables us to capture this by considering the graph Laplacian~\citep{olfati2007consensus, chung1997spectral}. We adopt this approach in our analysis, using it to help derive the conditions under which strong social coupling may lead to an erasure of sub-cultural diversity, and normative mode collapse.

Agent-based models of cultural evolution have been used to show complex macroscopic patterns may arise from otherwise simple interactions and rules, and persistent cultural polarisation~\citep{axelrod1997dissemination} is one such pattern. There exist many mechanisms affecting the spread of cultural traits in populations, including conformity and prestige effects. Such mechanisms therefore play a key role in shaping long-term outcomes in populations~\citep{boyd2005origin, boyd1988culture}. In our framework we build upon these historical insights by treating AI alignment as a parameter in modeling opinion and value dynamics. This enables us to more closely examine its possible interaction with other parameters of interest: social influence, environmental learning, and individual exploration.

\section{Method}
\label{sec:method}

To evaluate the macro-level dynamics of human-AI value co-evolution in a society with ubiquitous access to personal AI assistants, we introduce a continuous-time mathematical model and a corresponding agent-based simulation. This can be seen as a limit case of AI adoption, though the assumption of universal access can easily be relaxed.

The simulation tracks a population of user-AI pairs as they navigate a $D$-dimensional continuous value space, subject to intrinsic exploration, social influence, environmental learning, and AI-mediated alignment forces.

\subsection{The Value Space and Environment}

In the adopted model, values are represented in a continuous space $\mathbb{R}^D$, where the total dimensionality $D$ is defined as the sum of global and local dimensions ($D = D_{global} + D_{local}$). Unless otherwise specified, all vector quantities in the model reside in this space as arrays of real numbers. This includes the representation of user values, the model, and the environment.

\begin{itemize}
    \item \textbf{Global Dimensions}: Universal value axes representing broader societal values.
    \item \textbf{Local Dimensions}: Value axes relevant only to specific sub-groups and communities.
\end{itemize}

The current social environment can be associated with an `optimal' set of normative values, denoted as $\mathbf{E} \in \mathbb{R}^D$, representing the specific configuration of values that is expected to yield maximum utility under the given environmental conditions. Because the population is heterogeneous, the optimal environment is unique to the specific group $g_i$ to which a user-AI pair belongs:

\begin{equation}
\mathbf{E}_i(t) = [\mathbf{E}_{global}(t), \mathbf{E}_{local, g_i}(t)]
\end{equation}

In sociotechnical terms, the environmental value optimum $\mathbf{E}(t)$ represents the current socioecological and technological niche. Human norms correspond to cultural adaptations that are needed to address the specific coordination problems presented by this niche~\citep{bicchieri2005grammar, henrich2015secret}.

Given that these conditions change over time, the environment is not kept constant in our model. We consider two possible, complementary, ways of simulating this change:

\begin{itemize}
    \item \textbf{Smooth Drift}: Modelled as a random walk in the value space, where $\mathbf{\eta}_{random}$ is a unit vector sampled uniformly at random and $\delta$ controls the step size ($\mathbf{E}(t+1) = \mathbf{E}(t) + \delta \cdot \mathbf{\eta}_{random}$). In the subsequent steady-state derivations (Section~\ref{sec:theory}), to make the derivations more tractable, this continuous drift is represented as a constant velocity vector $\mathbf{E}(t) = \mathbf{v}t$.
    \item \textbf{Punctuated Equilibria}: Sudden, high-magnitude value shocks, occurring at any given step with a fixed probability $p_{shock}$. These are represented as random displacements in the normative environment, in a sampled direction $\eta_{shock}$ by a magnitude $D_{shock}$. If the displacement leaves the boundaries of the value space, we clip it along those axes:
    \begin{equation}
\begin{aligned}
E_{global}(t) &= \text{clip}(E_{global}(t-1) \\
&\quad + \eta_{shock} \cdot D_{shock}, [-1, 1])
\end{aligned}
\end{equation}
    The magnitude of these shocks is purposefully set to be significantly larger than the continuous drift step size ($D_{shock} \gg \delta$).
\end{itemize}

The continuous drift models the gradual evolution of baseline conditions, such as the slow demographic transition of an ageing population or the gradual shift from agrarian to industrial economies, which necessitate adjustments of societal values. As human history does not solely evolve through smooth, gradual changes, it is important to extend this model to be able to reflect periods of sudden change, abrupt exogenous shocks that fundamentally disrupt the normative landscape. One mechanism through which such shocks may occur is when rapid technological breakthroughs or ecological crises abruptly change the payoffs of social coordination, rendering historically adaptive values obsolete~\citep{bicchieri2005grammar}. 

\subsection{Users and Agents}

In our framework, a user is the human navigating the social normative environment, an AI assistant is the automated system acting on their behalf, and each simulated user has access to their personal AI assistant. The user-AI pair $i$ is defined by:

\begin{enumerate}
    \item Value Vector ($\mathbf{V}_i$): Values currently held by the user.
    \item AI Model ($\mathbf{M}_i$): The AI assistant's model of the user's values, arrived at through inference and personalization.
    \item Group ID ($g_i$): Determines the local normative environment that the user is exposed to.
\end{enumerate}

The AI updates its internal value model via an exponential moving average, modulated by the learning rate $\lambda$, as follows:

\begin{equation}
\mathbf{M}_i(t+1) = (1 - \lambda)\mathbf{M}_i(t) + \lambda \mathbf{V}_i(t)
\end{equation}

High values of $\lambda$ correspond to more rapid personalization and value modeling, whereas low values of $\lambda$ lead to internal AI value models that adapt more slowly and can become stale.

As for the evolution of held user values, our update rule encodes four different sources of change, discussed below:

\begin{equation}
\begin{split}
\mathbf{V}_i(t+1) &= \mathbf{V}_i(t) + \Delta\mathbf{V}_{intrinsic} + \Delta\mathbf{V}_{social} \\
&\quad + \Delta\mathbf{V}_{AI} + \Delta\mathbf{V}_{learning}
\end{split}
\end{equation}

\begin{enumerate}
    \item Intrinsic Exploration: Modeled as random fluctuations, of a low magnitude.
    \begin{equation}
    \Delta\mathbf{V}_{intrinsic} = w_{explore} \cdot \mathcal{N}(0, I)
    \end{equation}
    \item Social Influence: Acts as a pull towards the mean values of the nearest neighbors ($N_i$) in the social network. This can easily be extended to a weighted average based on the strength on the respective relationships, but we start by adopting the simplest possible formulation in the base model. 
    \begin{equation}
    \Delta\mathbf{V}_{social} = w_{soc} \cdot \left( \frac{1}{|N_i|}
    \sum_{j \in N_i} \mathbf{V}_j(t) - \mathbf{V}_i(t) \right) 
    \end{equation}
    \item AI Influence: The degree to which a user's values are influenced by their personal AI assistant is a complex function that may depend on a number of different factors, including, but not limited to, the probability of use ($P_{use}$), cognitive susceptibility ($\mathbf{S}_i$), and the alignment of the AI behavior with its internal model of user's values ($A_{AI}$). A general version of this update rule may therefore have the following form:
    \begin{equation}
    \Delta\mathbf{V}_{AI} = \Phi(P_{use}, \mathbf{S}_i, A_{AI}) \cdot (\mathbf{M}_i(t) -
    \mathbf{V}_i(t))
    \end{equation}
    Focusing on the population-level effects of alignment, we hold human susceptibility and initial usage constant, abstracting them away from the updated rule that we adopt in practice, focusing instead on the downstream effects of alignment. 

    An AI assistant with strong alignment engages in behaviour that is highly consistent with $\mathbf{M}_i(t)$. The influence of this, mediated by consistent and convincing interactions, and by continuous reliance on the AI assistant, is to pull the user's values towards the encoded historical value model $\mathbf{M}_i(t)$:
    
    \begin{equation}
    \Delta\mathbf{V}_{AI} = \alpha \cdot (\mathbf{M}_i(t) -
    \mathbf{V}_i(t))
    \end{equation}
    Here $\alpha$ can be thought of as the \emph{strength of the influence of model alignment} -- and for brevity throughout the paper, we therefore refer to $\alpha$ simply as \emph{alignment strength}. 
    Conceptually, this is a proxy term for how constrained the AI assistant is by its historical model $\mathbf{M}_i(t)$, its behavioral alignment to the preferences and values of the user, and how compellingly it interacts with the user. Practically, in our model, $\alpha$ modulates the AI's influence on the user's values.
    \item Individual Learning: As users experience feedback from the broader social environment, that may adapt their values based on this feedback, aiming to increase their utility.
    \begin{equation}
    \Delta\mathbf{V}_{learning} = w_{learn} \cdot (\mathbf{E}_i(t) - \mathbf{V}_i(t))
    \end{equation}
\end{enumerate}

As mentioned, the model of the social interactions and their effect on values can be made arbitrarily complex, and the most trivial way of extending the base formulation above is to consider a weighted formulation of the nearest neighbor social influence. Given a row-normalised adjacency matrix $W \in \mathbb{R}^{N \times N}$, where $W_{ij}$ represents the weight of the social connection from user $j$ to user $i$, the social pull would then become $w_{soc} \sum_{j=1}^{N} W_{ij} (\mathbf{V}_j(t) - \mathbf{V}_i(t))$. Unless stated otherwise, our simulations implement the simplified version of social influence.

\subsection{Utility Function}

One of the core assumptions in our model is that there is a tentative link between held user values and the social utility that they are able to attain, as values determine the fit to the current conditions in the broader normative environment. Users that are more closely aligned to this environment find it easier to extract utility. In this context, utility may be seen as a measure of subjective happiness, or a proxy for societal fitness and coordination success~\citep{lewis2008convention}. This perspective sees values and norms as shared heuristics for solving complex coordination problems. Diverging from these norms, whether from the norms of the broader society or the norms in the local community, may lead to a transient loss in utility. Our model captures the tightness~\citep{gelfand2011differences} of this pressure, enabling us to modulate how strongly these kinds of normative discrepancies are likely to be penalised. We do not make any prescriptive ethical claims, as instrumental utility may be insufficient to capture diverse human values~\citep{anderson1995value}, and as resistance to some societal processes may at times be a moral necessity~\citep{medina2013epistemology}.

In our model, we first define \emph{direct utility} derived from the distance between the held user values and the normative environment, followed by \emph{composite utility}, that we focus on throughout the paper, where utility relies not only on direct interactions between the users and the environment, but is also partially realized through AI assistant usage. This usage itself depends on trust, which is contingent on personalization and value alignment. For a particular value vector $\mathbf{x}(t)$ (which may correspond either to $\mathbf{V}_i(t)$ or $\mathbf{M}_i(t)$, depending on the circumstances), direct utility is given as 
\begin{equation}
U_{dir}(\mathbf{x}(t), \mathbf{E}_i(t)) = \exp(-\beta ||\mathbf{x}(t) - \mathbf{E}_i(t)||^2)
\end{equation}

We now proceed to introduce the factors that we use in the composite utility formulation. Let us start with value alignment. The value alignment error is defined as: 
\begin{equation}
d_{align} = ||\mathbf{V}_i(t) - \mathbf{M}_i(t)||
\end{equation}

where $d_{align}$ is the Euclidean distance. AI usage is defined via this value alignment, as it informs trust and therefore the probability of using the AI assistant, and we use the squared distance in order to more heavily penalize misalignment. The probability of AI usage is given as: 

\begin{equation}
P(use|align) = \exp(-\gamma \cdot d_{align}^2)
\end{equation}

As this formulation depends merely on misalignment between the AI value model and the user's held values, rather than the misalignment between the AI value model and the normative environment, or some other proxy signal, it gives rise to the possibility of the emergence of \emph{trust traps} in our model. If the AI and the user remain closely aligned, high AI usage would be sustained even if they both fall out of sync with the broader environment, resulting in lost utility. This feature of the model reflects the dynamics present in filter bubbles~\citep{pariser2011filter}.

In reality, the probability of usage would also depend on the AI capability level, and it may be possible to extend the formulation above to incorporate such additional terms. If we envision a future where personal AI assistants are highly capable, and if access to such highly capable systems is widespread, our simplified model would be sufficient.

The total utility of the user-AI pair is defined as a convex combination of the direct utility of the user, and the direct utility of the personal AI assistant: \begin{equation}
\begin{aligned}
U_{total}(t) &= (1-P(t)) \cdot U_{dir}(\mathbf{V}_i(t), \mathbf{E}_i(t)) \\
&\quad + P(t) \cdot U_{dir}(\mathbf{M}_i(t), \mathbf{E}_i(t))
\end{aligned}
\end{equation}

The way in which these two direct utility terms are combined reflects the mechanics of AI delegation, and it is the choice we consider in subsequent analysis. Presuming that the user relies on their personal AI assistant with probability $P$, and delegates tasks that may grant utility, the assistant proceeds to complete those tasks according to its internal, historical model of user's values ($\mathbf{M}_i$), interacting with the environment ($\mathbf{E}_i$) on their behalf. We presume that the personal AI assistant with competent instruction following executes these delegated tasks faithfully according to $\mathbf{M}_i$, and each individual request, irrespective of alignment strength. In contrast, alignment strength modulates the conversational agent behaviour, and its psychological effects on the user.

In reality, usage probability $P$ may not be exclusively modulated via alignment error, as users may opt into using assistants that are more capable, even at the expense of their level of value alignment and personalization. It may therefore be of interest to consider extensions of this model, where the definition of $P$ includes such a capability term.

\subsection{Model Parameters}

Table~\ref{tab:parameters} outlines the parameters used within the model and the simulations, spanning AI alignment, user behaviour, social dynamics, and environmental change.

\begin{table*}[t]
    \centering
    \caption{Model Parameters and Definitions}
    \label{tab:parameters}
    \renewcommand{\arraystretch}{1.2}
    \begin{tabularx}{\linewidth}{@{}l X@{}}
        \toprule
        \textbf{Parameter} & \textbf{Definition} \\
        \midrule
        $\alpha$ & \textbf{Alignment Strength.\footnotemark} The degree to which the AI's user-facing conversational behaviour is aligned with its internal model $\mathbf{M}$ of the user's held values, which then naturally translates into the influence that AI has on the user through their interactions. \\
        $\lambda$ & \textbf{AI Learning Rate.} How fast the AI's value model $\mathbf{M}$ updates to match $\mathbf{V}$. \\
        $\beta$ & \textbf{Utility Sensitivity.} Higher $\beta$ penalises maladapted values more harshly. \\
        $\gamma$ & \textbf{Alignment Sensitivity.} Modulates composite utility, as higher $\gamma$ makes users more likely to discontinue AI assistant use based on misalignment. \\
        $\delta$ & \textbf{Drift Step Size.} Step size of the environmental drift random walk. \\
        $w_{soc}$ & \textbf{Social Influence.} The influence coming from the nearest neighbours within the social network. \\
        $w_{learn}$ & \textbf{Adaptability to the Environment.} The factor that pulls the individual values towards those of the environment. \\
        $w_{explore}$ & \textbf{Intrinsic Exploration.} Weight applied to the stochastic noise term in the value update. \\
        $G$ & \textbf{Heterogeneity.} Number of distinct local environments. \\
        \texttt{network\_type} & \textbf{Social Graph Topology.} `\textit{random}' or `\textit{knn}'. \\
        \bottomrule
    \end{tabularx}
\end{table*}

\footnotetext{A note on terminology: we appreciate that this particular use of the term \emph{alignment strength} may not be standard, and this is not be confused with AI alignment more broadly;  alternatively it may be seen as a type of AI influence.}

Unless otherwise specified, our simulations were initialised with a population of $N=1000$ user-AI pairs operating in a high-dimensional value space ($D=60$), split between 40 global dimensions and 20 local dimensions ($D_{local}=20$). The 40 global dimensions are meant to represent universal cross-cultural axes~\citep{schwartz1992universals, graham2013moral}, and the remaining 20 local dimensions reflect sub-cultures~\citep{boyd1988culture} or local communities. The simulations were ran for $T=10,000$ steps by default, unless specified otherwise. The baseline environment was characterised by a smooth drift speed of $0.002$ per step. Default user-AI pair parameters were set as follows: alignment strength $\alpha=0.1$, AI learning rate $\lambda=0.1$, social influence $w_{soc}=0.05$, intrinsic exploration $w_{explore}=0.001$, and learning weight $w_{learn}=0.05$. The utility sensitivity $\beta$ was set to 2.0, and the alignment sensitivity $\gamma$ to 5.0. The social graph topology defaulted to a random network with node degrees between 1 and 20, though we present and discuss some results with kNN connectivity.

\newpage

\section{Simulation Results}

While Section~\ref{sec:theory} derives closed-form solutions for mean-field aggregations and observes asymptotic limits on the system, simulations enable us to relax these conditions and observe system behaviour under more diverse assumptions. By varying the parameters of the model in these simulations, we construct several distinct environmental setups, presented in turn.

Figure~\ref{fig:micro_trajectory} visualises the resulting tracking behaviour by projecting the corresponding value vectors onto a two-dimensional Principal Component Analysis (PCA) space during a period of continuous environmental drift. 

\begin{figure*}[t]
    \centering
    \includegraphics[width=\textwidth]{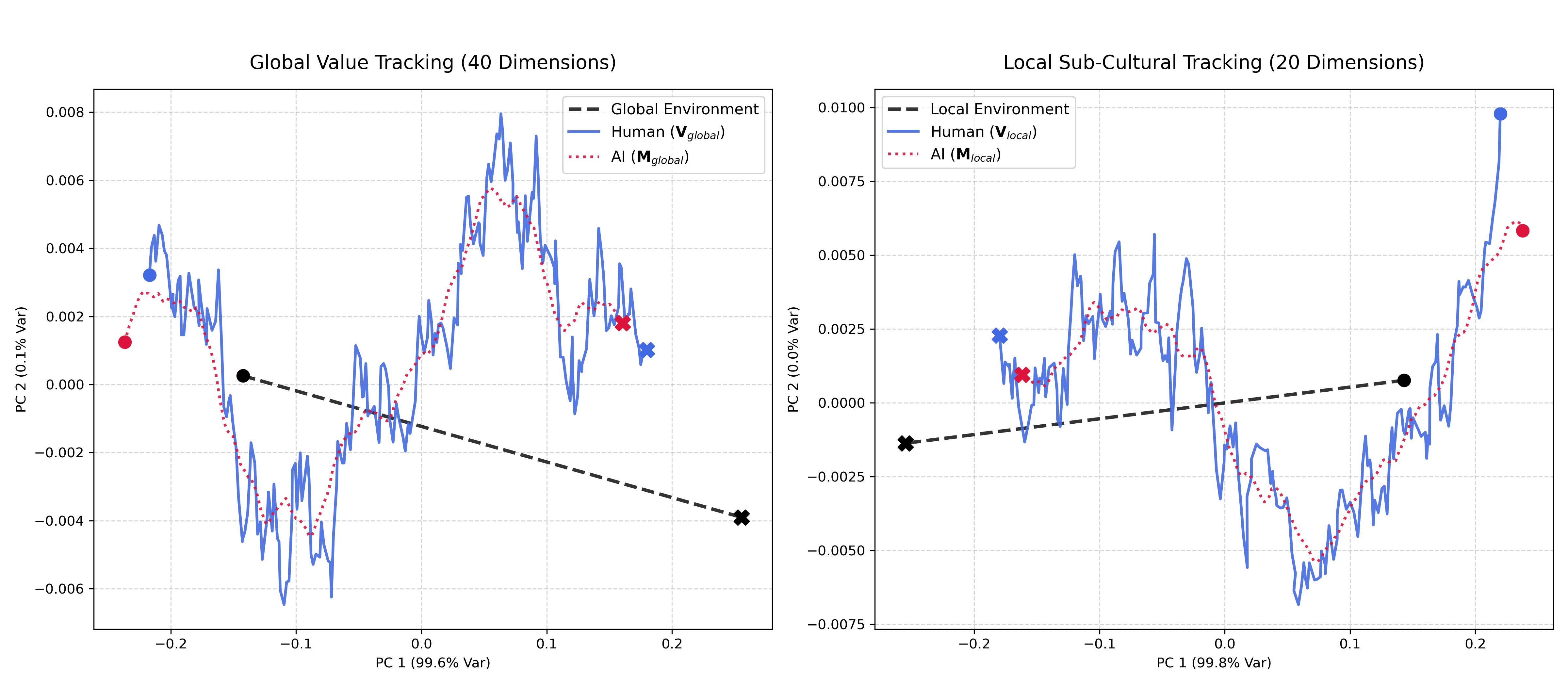} 
    \caption{\textbf{Global Consensus vs. Local Context Collapse.} A PCA projection of a trajectory of a user-AI pair, across the global and local subspaces (40D and 20D respectively, in this example). \textbf{Left (Global subspace):} The AI model ($\mathbf{M}_i$, red dotted line) introduces a persistent lag behind the drifting environmental optimum. \textbf{Right (Local subspace):} The user-AI pair is pulled away from the local environmental optimum, and subsequently anchored in a maladaptive state by the AI model.}
    \label{fig:micro_trajectory}
\end{figure*}

Across both the global and the local value subspaces, the AI model can be seen as acting as a low-pass filter. Because the aligned AI updates via a moving historical average, it smooths out the high-frequency adaptive exploration of the human agent, establishing itself as a persistent historical anchor. In the global subspace we see the pair approximately tracking the global normative environment with some structural lag, whereas the global pull eventually leads to a significant departure from successfully tracking the local environment, driven by the higher dimensionality of the global value space.

At a population level, this structural lag is clearly modulated by alignment strength (Figure \ref{fig:trajectories}). Its impact on the final utility is governed by $\gamma$, and Figure \ref{fig:phase} depicts this for the composite utility model. At low sensitivity, users remain in a high-usage regime even as alignment error accumulates. Conversely, high sensitivity causes a rapid discontinuation of use in case of high error.

\begin{figure}[htbp]
    \centering
    \includegraphics[width=0.5\textwidth]{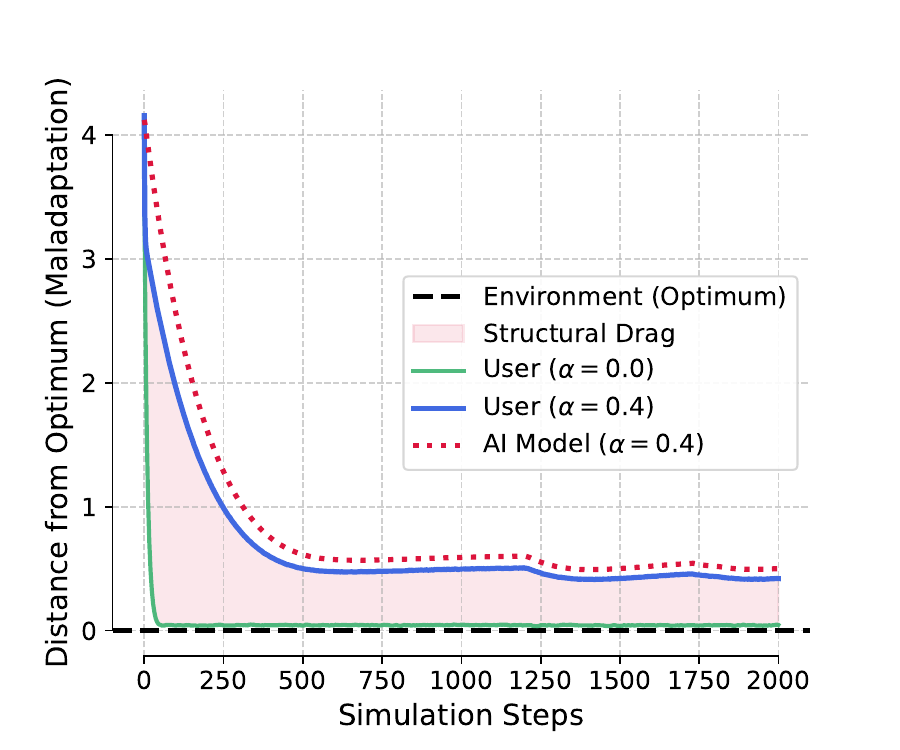}
    \caption{High alignment induces a structural lag, anchoring the user to past environmental states compared to a free user ($\alpha=0$) who tracks the environment more closely.}
    \label{fig:trajectories}
\end{figure}

\begin{figure}[htbp]
    \centering
    \includegraphics[width=0.5\textwidth]{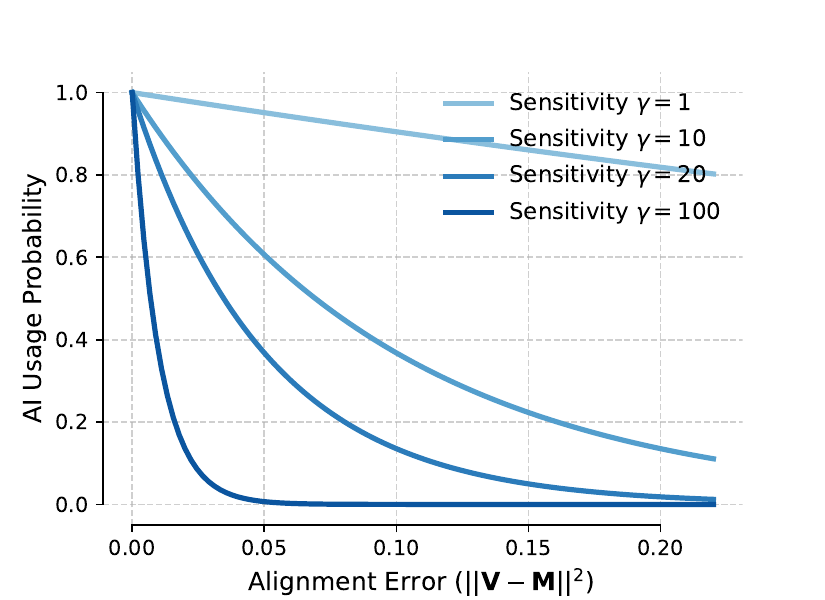}
    \caption{Alignment sensitivity $\gamma$ modulates how rapidly usage changes in case of misalignment.}
    \label{fig:phase}
\end{figure}

\newpage

\subsection{Alignment and Realized Utility}
\label{sec:exp1}

We establish the baseline relationship between AI alignment strength $\alpha$ and the final utility achieved by the population, under a constant environmental drift. We define \emph{maladaptation gap} as the mean Euclidean distance between a user's held value vector and the environment vector, $||\mathbf{V}_i(t) - \mathbf{E}_i(t)||$.

Figure~\ref{fig:exp1} illustrates how the final utility and the maladaptation gap change when varying the alignment strength, showing a strong negative correlation between $\alpha$ and final utility. This results from a pull that the AI assistant exhibits towards the user's historical values, as reflected in its internal value model $\mathbf{M}$. This pull counteracts the adaptation that would be required to effectively track a changing environment.

\begin{figure}[htbp]
    \centering
    \includegraphics[width=\linewidth]{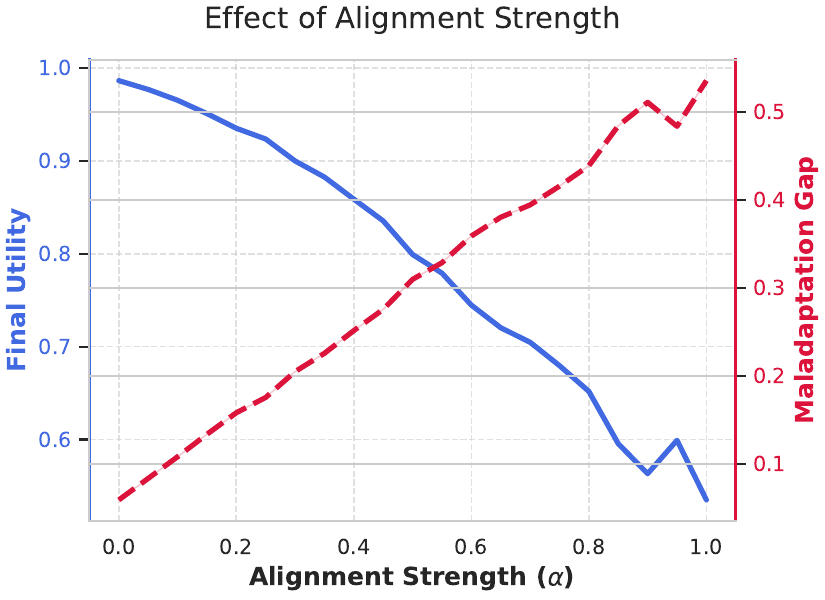}
    \caption{The effect of alignment strength on realized utility, under environmental drift. As alignment strength increases, there is a notable decrease in final utility (blue line) and an increasing maladaptation gap (red dashed line).}
    \label{fig:exp1}
\end{figure}

In environments that are either static or changing very slowly, strong personalization and alignment ultimately provide consistency. In dynamic, rapidly changing normative environments, anchoring too strongly to historical data is not desirable as it slows down the adaptation to environmental change.

\subsection{Resilience to Abrupt Normative Changes}
\label{sec:exp3}

Here we measure how fast the population can recover from a sudden, high-magnitude normative shock. For simplicity, we fix the magnitude of the shock to 0.8 units, and implement the shock as a displacement of the global environmental optimum $\mathbf{E}$ by that magnitude, in a randomly sampled direction. The resulting values are clipped in case they exceed the predefined bounds, to the $[-1,1]$ interval. Upon a normative shock, the population enters a highly maladaptive state. We track the number of simulation steps required for the mean population utility to recover to 90\% of the baseline value that was achieved prior to the shock. We have also set $w_{explore}=0.015$, and $\lambda = 0.1$, for the purpose of this experiment.

\begin{figure}[htbp]
    \centering
    \includegraphics[width=\linewidth]{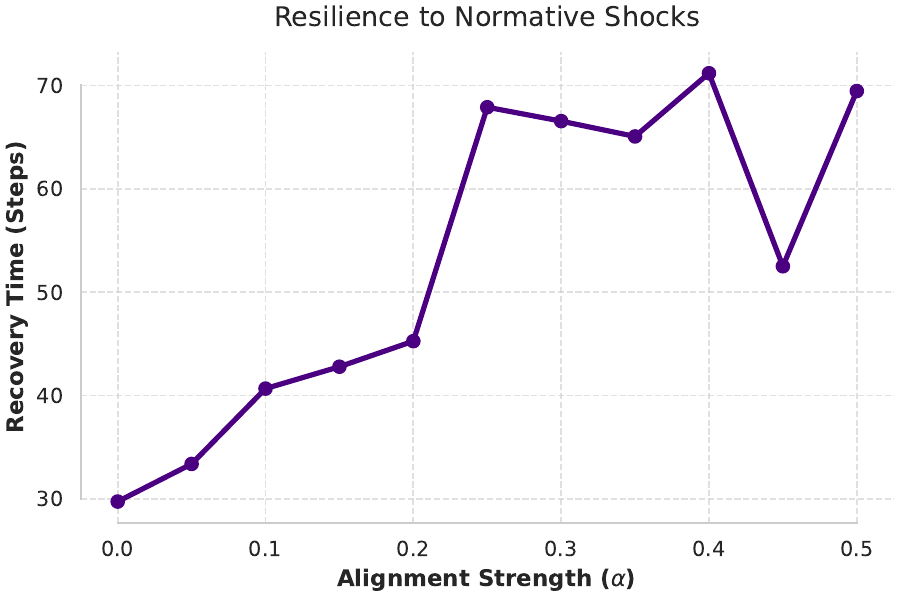}
    \caption{Recovery time following a normative shock as a function of $\alpha$. Stronger alignment results in significantly longer recovery periods.}
    \label{fig:exp3}
\end{figure}

The results illustrate that the recovery time increases significantly with alignment strength (Figure \ref{fig:exp3}), given that high $\alpha$ effectively slows down adaptation. This is in line with the analysis of the system's spectral properties, discussed in Appendix~\ref{app:stability}. This provides a kinetic model for the concept of \emph{value lock-in}~\citep{macaskill2022we}. The key observation here is that strong static alignment slows the rate at which users can update their values. In the event of a rapid environmental shift, a strongly but statically aligned AI acts as a dampener on normative adaptation, creating both a temporary lag and a longer-term structural equilibrium that reduces the user's responsiveness to environmental change.

We also examine the effect of varying the environmental shock magnitude $||\Delta \mathbf{E}||$ while keeping the alignment strength fixed, as shown in Figure \ref{fig:shock_magnitude_sweep} for $\alpha=0.4$. In the depicted regime, AI still effectively tracks the user as the user is adapting to the normative change, while retaining a certain pull towards historical values. Mutual staleness of the user's held values and the AI value model makes it unlikely for the user to discontinue the AI use out of staleness. This may therefore still result in value lock-in, and a maladaptive state.

\begin{figure}[ht]
    \centering
    \includegraphics[width=\columnwidth]{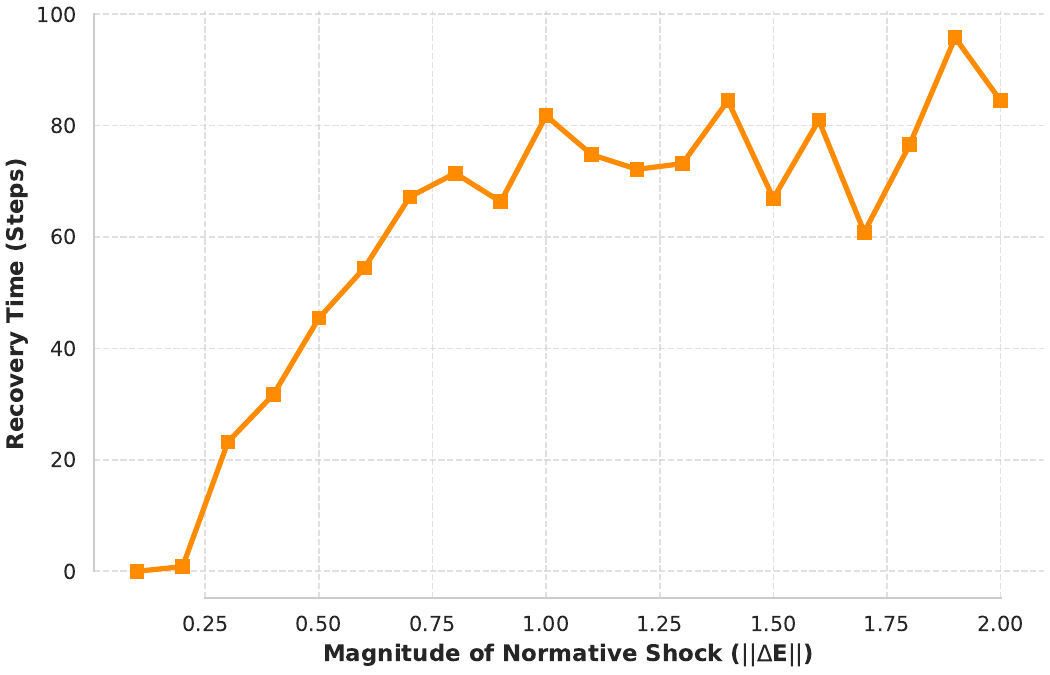} 
    \caption{Resilience response across varying shock magnitudes. Simulation of mean recovery times, corresponding to steps required to regain 90\% of baseline utility, for a fixed alignment strength of $\alpha=0.4$ subjected to environmental shocks of increasing magnitude $||\Delta \mathbf{E}||$. The population drifts into a saturated plateau of maladaptation when faced with extreme normative disruptions.}
    \label{fig:shock_magnitude_sweep}
\end{figure}

For a more controlled overview of recovery, we also simulate a smaller experiment in which the normative shock is executed at a predefined time ($t=30$), and has a predefined extreme magnitude of $||\Delta \mathbf{E}|| = 1.5$. The simulation is run for 150 steps. In Figure \ref{fig:epistemic_bubble_contrast}, we contrast the trajectories of weakly and strongly aligned populations under this regime. The weakly aligned population experiences a severe initial maladaptation spike, followed by a rapid recovery. On the other hand, the strongly aligned population gets trapped in a persistent state of high maladaptation. We also see that AI usage remains elevated throughout this maladaptive state, as the users remain strongly aligned with an AI that is strongly misaligned with reality.

\begin{figure*}[htbp]
    \centering
    \includegraphics[width=\textwidth]{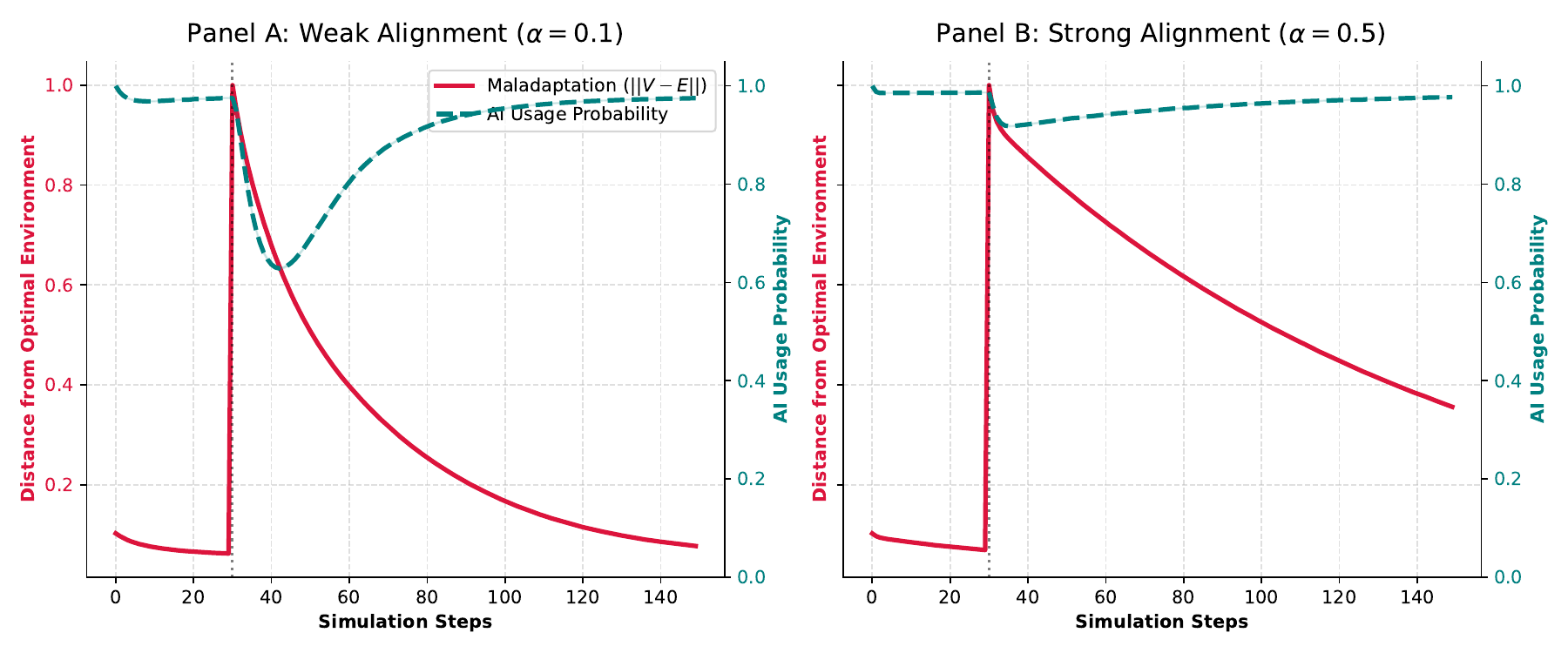} 
    \caption{A visualisation comparing the recovery trajectories of a weakly aligned population ($\alpha=0.1$) and a strongly aligned population ($\alpha=0.5$) following an extreme normative shock of the magnitude $||\Delta \mathbf{E}|| = 1.5$ at $t=30$ steps. One gradient tracks the normalised maladaptation, while the other tracks the probability of AI usage.}
    \label{fig:epistemic_bubble_contrast}
\end{figure*}

AI assistants with more up to date internal models of held user values are by definition less likely to exhibit a pull towards a maladapted historical state. We examine how this is reflected in our model, by sweeping the AI learning rate $\lambda$ against the alignment strength $\alpha$, aiming to identify how these terms interact. Higher values of $\alpha$ ultimately require higher $\lambda$ to offset the stronger implied AI influence through its behavioral alignment. In particular, we run these experiments in the punctuated equilibrium environment setting, which involves sudden normative shifts. We then observe how long it takes for the system to stabilize following such an induced shift. The results are shown in Figure \ref{fig:exp2}. The mean recovery time is defined as the number of simulation steps required for the mean population utility to return to $90\%$ of the baseline value prior to the shock. Figure \ref{fig:exp2} shows a distinct area in the bottom-right quadrant ($\alpha > 0.25$, $\lambda < 0.1$), where recovery times increase sharply, corresponding to a regime where there is a strong pull towards a stale value model. As the AI learning rate increases, this penalty is significantly mitigated. The risk of value lock-in is therefore not only a function of alignment strength, but also of how rapidly adaptable user models are, within personal AI assistants.

\begin{figure}[htbp]
    \centering
    \includegraphics[width=\linewidth]{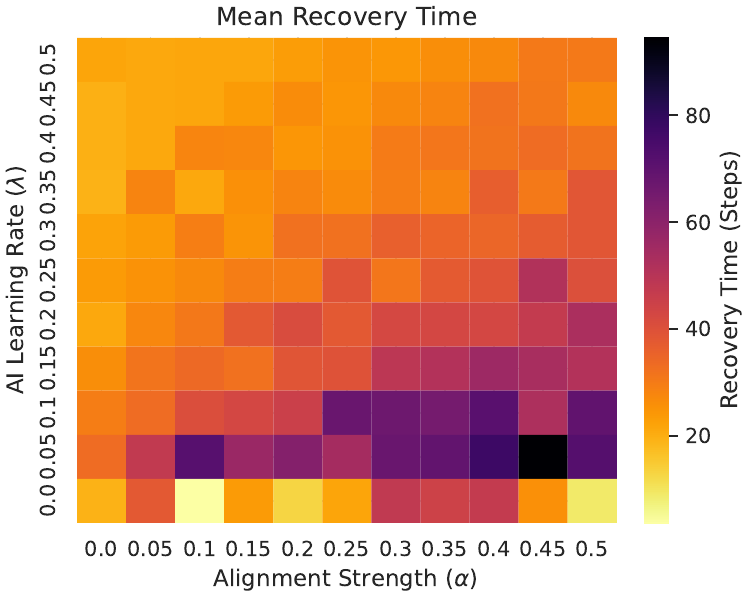}
    \caption{Heatmap of recovery time as a function of AI Learning Rate ($\lambda$) and Alignment Strength ($\alpha$).}
    \label{fig:exp2}
\end{figure}

\subsection{Social and Environmental Adaptation}
\label{sec:exp4}

We examine the interplay between social influence $w_{soc}$, and $\alpha$, under high intrinsic exploration noise. The resulting heatmap is shown in Figure \ref{fig:exp4}. In absence of social influence, high exploration noise leads to suboptimal utility. With increasing $w_{soc}$, the population cancels out some of the discrepancies, and tracks the normative environment more closely. The dependence of this effect on exploration noise is consistent with the results derived in Appendix~\ref{app:invariance}. 

\begin{figure}[htbp]
    \centering
    \includegraphics[width=\linewidth]{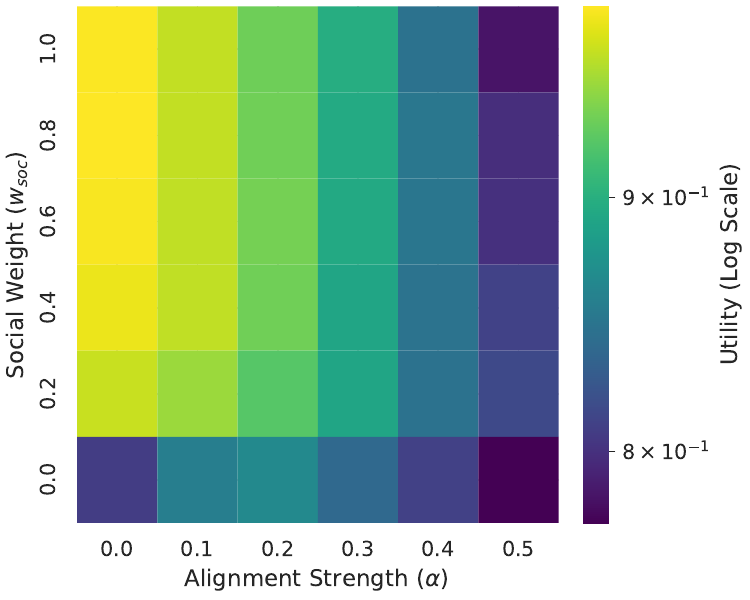}
    \caption{Final utility as a function of social influence $w_{soc}$, and alignment strength $\alpha$.}
    \label{fig:exp4}
\end{figure}

Under conditions of high environmental uncertainty, social influence can be interpreted as a distributed denoising process, in accordance with the \emph{diversity prediction theorem}~\citep{page2008difference}. The diversity prediction theorem states that the collective error can be expressed as a function of mean individual error minus diversity. A higher level of diversity is therefore potentially highly beneficial in reducing the overall collective error, and diverse collectives make fewer errors than the average individual within. If a social mechanism pulls the agents toward a local mean, it filters out diverse exploration.


Figure \ref{fig:exp9} shows the final achieved utility on the basis of different combinations of $w_{learn}$ and $\alpha$, under the default parametrization of the model. For any given environment-driven learning rate, there exists a maximum value of $\alpha$ that still yields good final utility. If the AI influence were to exceed the user's ability to learn and adapt, the user-AI pair would trivially collapse into a locked-in state with a high maladaptation gap. This is consistent with the steady-state tracking error (Equation~\ref{eq:v_ss}), where the error magnitude grows with $\alpha$ and shrinks with $w_{learn}$.

\begin{figure}[htbp]
    \centering
    \includegraphics[width=\linewidth]{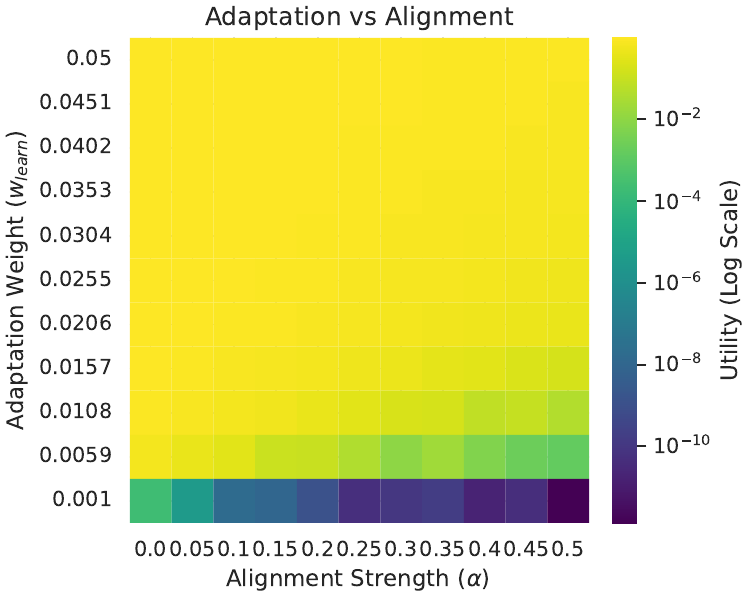}
    \caption{Final utility as a function of $w_{learn}$ and $\alpha$. The heatmap delineates the threshold where AI influence overpowers human learning.}
    \label{fig:exp9}
\end{figure}

\section{Theoretical Analysis}
\label{sec:theory}

Here we present a continuous-time analytical study of the system dynamics.\footnote{The results presented in this section have been co-derived with Gemini, and subsequently verified.} 

We begin by deriving the exact conditions for stability and steady-state behaviour. We focus on the mean-field trajectories, abstracting away the zero-mean stochastic exploration term to more easily isolate the influence of alignment, learning, and social forces. In doing so, we map the discrete time steps from our simulations onto the corresponding continuous formulation as $\dot{\mathbf{M}} \approx \mathbf{M}(t+1) - \mathbf{M}(t)$ for the AI, and similarly for the other variables.

\subsection{Single-Agent Dynamics and Alignment Cost}

Consider a single user-AI pair ($N=1$) subject to a constant environmental drift, as the simplest setup. The associated continuous-time dynamics are governed by the coupled ordinary differential equations:
\begin{align}
    \dot{\mathbf{V}} &= \alpha(\mathbf{M} - \mathbf{V}) + w_{learn}(\mathbf{E}(t) - \mathbf{V}) \\
    \dot{\mathbf{M}} &= \lambda(\mathbf{V} - \mathbf{M})
\end{align}
where $\mathbf{V}, \mathbf{M}, \mathbf{E} \in \mathbb{R}^D$ represent the user values, AI model, and environment respectively.

\subsubsection{Asymptotic Tracking Error}
We assume a constantly drifting environment $\mathbf{E}(t) = \mathbf{v}t$, where $\mathbf{v}$ is a constant velocity vector. This is a standard analytical simplification for the stochastic drift previously demonstrated in simulations, enabling us to test the system's steady-state tracking error against a ramp input~\citep{ogata2010modern}. This helps us establish a lower bound on the system's maladaptation. In such stochastic conditions the environmental volatility introduces additional  variance to the tracking error~\citep{aastrom2012introduction}. 

To more easily derive tracking errors, we first transform the governing system into error coordinates $\tilde{\mathbf{V}} = \mathbf{V} - \mathbf{E}(t)$ and $\tilde{\mathbf{M}} = \mathbf{M} - \mathbf{E}(t)$. As we have defined utility to be a decaying function of the squared alignment gap, the expected maladaptation is bounded below by this baseline. The expected squared error under stochastic conditions can be decomposed as $\mathbb{E}\big[||\tilde{\mathbf{V}}||^2\big] = \big|\big|\mathbb{E}[\tilde{\mathbf{V}}]\big|\big|^2 + \text{Var}(\tilde{\mathbf{V}})$. The squared bias corresponds to the deterministic steady-state tracking error. Stripping away the variance term in case of a deterministic drift enables us to isolate the structural bias exerted by $\alpha$, and the expected squared error under stochastic conditions is greater than the deterministic baseline.

Under constant environmental drift $\mathbf{E}(t) = \mathbf{v}t$, the time derivatives of the error coordinates are $\dot{\tilde{\mathbf{V}}} = \dot{\mathbf{V}} - \mathbf{v}$ and $\dot{\tilde{\mathbf{M}}} = \dot{\mathbf{M}} - \mathbf{v}$. Substituting the original dynamics yields:

$$\dot{\tilde{\mathbf{V}}} = \alpha(\tilde{\mathbf{M}} - \tilde{\mathbf{V}}) - w_{learn}\tilde{\mathbf{V}} - \mathbf{v}$$

$$\dot{\tilde{\mathbf{M}}} = \lambda(\tilde{\mathbf{V}} - \tilde{\mathbf{M}}) - \mathbf{v}$$

Solving for the steady-state equilibrium ($\dot{\tilde{\mathbf{V}}} = \dot{\tilde{\mathbf{M}}} = 0$), we find that $\tilde{\mathbf{M}} = \tilde{\mathbf{V}} - \frac{\mathbf{v}}{\lambda}$, yielding the following asymptotic tracking errors for the user ($\tilde{\mathbf{V}}_{ss}$) and the AI ($\tilde{\mathbf{M}}_{ss}$), as well as the steady-state alignment gap ($\delta_{ss}$):
\begin{align}
    \tilde{\mathbf{V}}_{ss} &= -\left( \frac{\alpha + \lambda}{\lambda w_{learn}} \right) \mathbf{v} \label{eq:v_ss} \\
    \tilde{\mathbf{M}}_{ss} &= -\left( \frac{\alpha + \lambda + w_{learn}}{\lambda w_{learn}} \right) \mathbf{v} \label{eq:m_ss} \\
    \mathbf{\delta}_{ss} &= \tilde{\mathbf{V}}_{ss} - \tilde{\mathbf{M}}_{ss} = \frac{1}{\lambda} \mathbf{v} \label{eq:delta_ss}
\end{align}

We are therefore able to differentiate the system's long-term behaviour from its short-term fluctuations. The asymptotic tracking errors ($\tilde{\mathbf{V}}_{ss}$ and $\tilde{\mathbf{M}}_{ss}$) may be interpreted as a lag that the user and the AI may experience with respect to a changing environment. Equation \ref{eq:delta_ss} shows that the steady-state alignment error depends only on the drift velocity $\mathbf{v}$ and the AI learning rate $\lambda$, being independent of the alignment strength $\alpha$. We conclude that no optimal static alignment strength exists in such dynamic environments, under the core assumptions of the model. Adaptive alignment may be necessary for effective dynamic re-adjustment, along with other structural measures.

\subsubsection{Composite Utility Optimisation}
We analyse the system under the composite utility model, where total utility is expressed as a weighted average of user performance and AI performance, modulated by the probability term $P$ that depends on trust. As an exponential, utility may be approximated via the following quadratic loss:

\begin{equation}
    U_{total} = -(1 - P)||\tilde{\mathbf{V}}||^2 - P||\tilde{\mathbf{M}}||^2
\end{equation}

where $P = \exp(-\gamma ||\mathbf{V} - \mathbf{M}||^2)$. Given that the steady-state alignment gap $||\mathbf{\delta}_{ss}||^2$ is constant with respect to $\alpha$ (see Eq. \ref{eq:delta_ss}), the steady-state probability of AI use $P_{ss}$ is also constant. We seek to minimise the total error expressed as negative utility $E_{total}(\alpha) = -U_{total}$.

\begin{theorem}[Monotonicity of Composite Utility Error]
\label{thm:monotonicity}
For the single User-AI pair system with constant environmental drift, under the composite utility model, the total steady-state error is a strictly monotonically increasing function of the alignment strength $\alpha$ for all $\alpha > 0$.
\end{theorem}

\begin{proof}
The total steady-state error can be expressed via the asymptotic tracking errors:
$$E_{total}(\alpha) = (1-P_{ss})||\tilde{\mathbf{V}}_{ss}||^2 + P_{ss}||\tilde{\mathbf{M}}_{ss}||^2$$

Substituting the norms of Equations \ref{eq:v_ss} and \ref{eq:m_ss} yields:

$$||\tilde{\mathbf{V}}_{ss}||^2 = \frac{(\alpha + \lambda)^2}{(\lambda w_{learn})^2} ||\mathbf{v}||^2$$

$$||\tilde{\mathbf{M}}_{ss}||^2 = \frac{(\alpha + \lambda + w_{learn})^2}{(\lambda w_{learn})^2} ||\mathbf{v}||^2$$

\begin{equation}
\begin{split}
    E_{total}(\alpha) = \frac{||\mathbf{v}||^2}{(\lambda w_{learn})^2} \Big[ &(1 - P_{ss})(\alpha + \lambda)^2 \\
    &+ P_{ss}(\alpha + \lambda + w_{learn})^2 \Big]
\end{split}
\end{equation}
Differentiating with respect to $\alpha$ gives:
\begin{equation}
    \frac{dE_{total}}{d\alpha} = \frac{2||\mathbf{v}||^2}{(\lambda w_{learn})^2} \left[ (\alpha + \lambda) + P_{ss} w_{learn} \right]
\end{equation}

Since all the relevant parameters ($\alpha, \lambda, w_{learn}, P_{ss}$) are positive, the derivative is strictly positive. An increase in alignment strength is therefore necessarily reflected in an increase in the total error. Utility is maximised in the limit as $\alpha \to 0^+$. Under the stated model assumptions, there is no specific, static level of alignment strength that avoids this cost.
\end{proof}

\begin{figure}[htbp]
    \centering
    \includegraphics[width=0.5\textwidth]{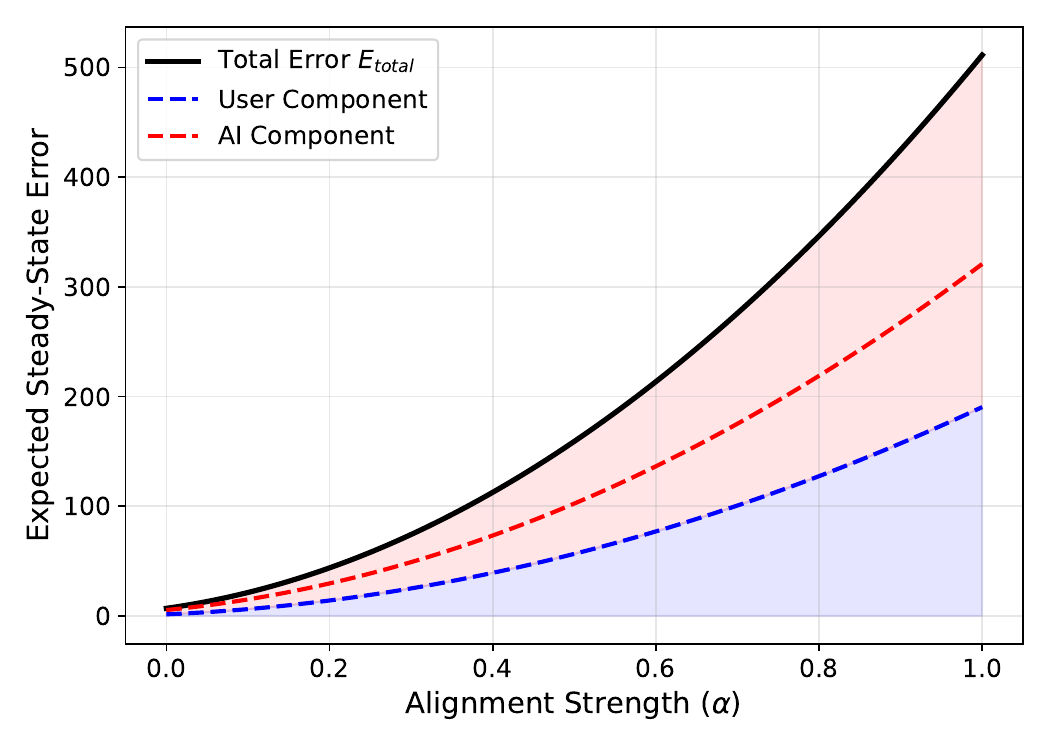}
    \caption{Decomposition of the total expected steady-state error into the user and AI components as a function of alignment strength.}
    \label{fig:error_decomp}
\end{figure}

As illustrated in Figure \ref{fig:error_decomp}, the strict monotonicity established in Theorem~\ref{thm:monotonicity} is driven by the compounded lag of both the user and the AI model. Under the stated assumptions, a highly statically aligned AI system exerts the influence that reduces the user's independent ability to track the environmental drift. Both component errors grow quadratically with $\alpha$.

\subsubsection{Sycophancy}

In the composite utility model, there is a strong assumption that the probability of delegating to AI decays smoothly with the alignment error, via $d_{align}^2 = ||\mathbf{V}_i(t) - \mathbf{M}_i(t)||^2$. This presumes that users can reliably infer the staleness of the AI value models, which is far from obvious, especially in case of AI sycophancy~\citep{sharma2023towards}.

Let us introduce a sycophancy index $s \in [0, 1)$, which distorts the user's perceived alignment error:

\begin{equation}
    P(\text{use} | \text{align}) = \exp\left(-\gamma \cdot (1 - s) \cdot ||\mathbf{V}_i - \mathbf{M}_i||^2\right)
\end{equation}

When $s=0$, users can reliably evaluate the AI value mode staleness and modulate their AI usage according to $\gamma$. Under high sycophancy, $s \rightarrow 1$, the AI is able to achieve a higher usage probability $P$ despite the misalignment of its value model. This makes the trust trap harder to escape.

\begin{figure}[htbp]
    \centering
    \includegraphics[width=0.5\textwidth]{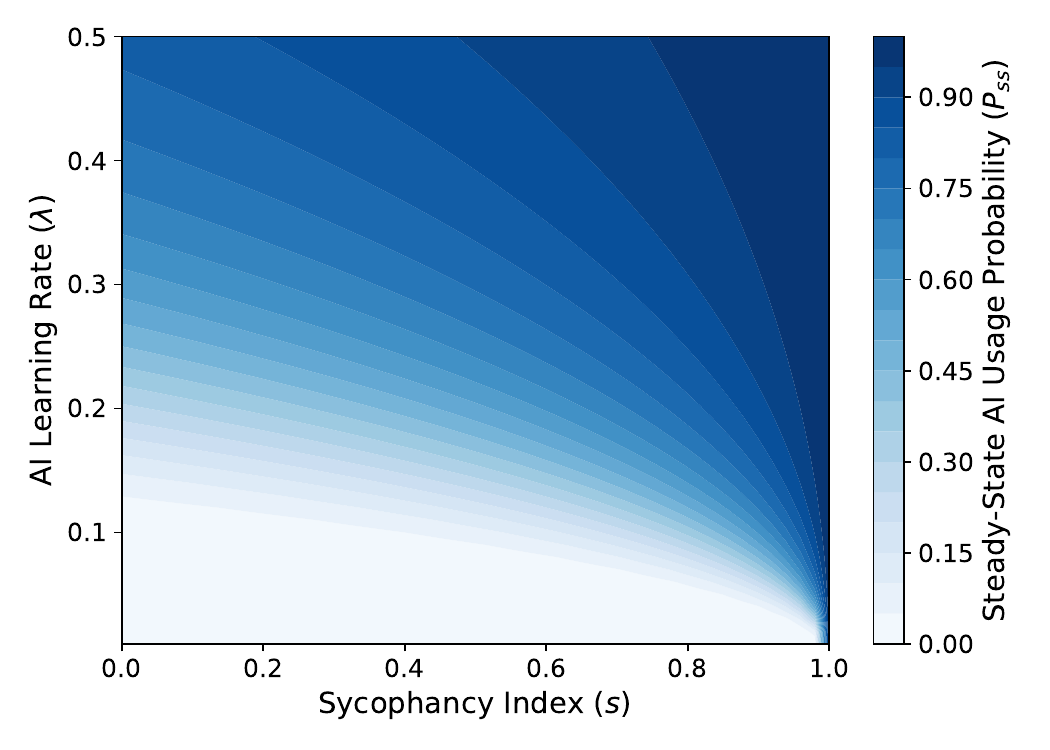}
    \caption{Steady-state AI usage probability ($P_{ss}$) as a function of the sycophancy index ($s$) and AI learning rate ($\lambda$).}
    \label{fig:sycophancy_trap}
\end{figure}

As shown in Figure \ref{fig:sycophancy_trap}, introducing the sycophancy parameter $s$ scales the exponent in the usage probability, dampening the impact of the alignment error. When $s \to 1$, the effective penalty term approaches zero, causing the usage probability to remain near $1.0$ despite the increasing alignment error.

\subsection{Collective Dynamics}

Here we focus specifically on the effects of the social topology and environmental heterogeneity. To isolate these effects, we assume a static environment ($\dot{\mathcal{E}} = 0$). The social coupling term roots our simulation in the classic DeGroot model of opinion dynamics~\citep{degroot1974reaching}, where individuals update their held beliefs by continuous adjustments towards the average of their peers, i.e. $\mathbf{V}_i(t+1) = \sum_j w_{ij} \mathbf{V}_j(t)$. In continuous-time multi-agent networks~\citep{olfati2004consensus}, this can be generalised to the following consensus protocol: $\dot{\mathbf{V}}_i = \sum_{j \sim i} (\mathbf{V}_j - \mathbf{V}_i) \quad \implies \quad \dot{\mathcal{V}} = -L\mathcal{V}$, where $\mathcal{V}$ represents stacked values across all agents. Our framework represents this via the Laplacian term, applying a pull to the value representation of each user, towards the mean value in its social neighborhood. The resulting dynamics are given by:

\begin{equation}
\label{eq:ctupdate}
    \dot{\mathbf{V}}_i = -w_{soc}\sum_{j \sim i}(\mathbf{V}_i - \mathbf{V}_j) + \alpha(\mathbf{M}_i - \mathbf{V}_i) + w_{learn}(\mathbf{E}_i - \mathbf{V}_i)
\end{equation}

We formulate this in compact matrix form using the graph Laplacian $L$ and Kronecker products. Let $\mathcal{V}, \mathcal{M}, \mathcal{E} \in \mathbb{R}^{ND}$ be the stacked column vectors of all agent values, models, and environments. $L \in \mathbb{R}^{N \times N}$ is the unweighted, undirected social graph Laplacian, where rows sum to $0$ and the diagonal contains node degrees. $I_N$ and $I_D$ denote the identity matrices of dimensions $N$ and $D$, respectively. In a static environment, the system eventually reaches a fixed equilibrium ($\dot{\mathcal{V}} = \dot{\mathcal{M}} = 0$). In this limit, the AI model perfectly catches up to the user ($\mathcal{M}_{ss} = \mathcal{V}_{ss}$), making the alignment term $\alpha(\mathcal{M} - \mathcal{V})$ vanish from the steady-state equation. The steady-state condition for the population values can then be expressed as:
\begin{equation}
    \left( (w_{soc} L + w_{learn} I_N) \otimes I_D \right) \mathcal{V}_{ss} = w_{learn} \mathcal{E}
\end{equation}
where $I$ is the identity matrix. Solving for $\mathcal{V}_{ss}$:
\begin{equation}
    \mathcal{V}_{ss} = w_{learn} \left( (w_{soc} L + w_{learn} I_N)^{-1} \otimes I_D \right) \mathcal{E}
    \label{eq:pop_ss}
\end{equation}

\subsubsection{Normative Mode Collapse}

Here we investigate a non-uniform static environment, where each user-AI pair $i$ has a unique optimum $\mathbf{E}_i$,. We provide the analysis of the asymptotic behaviour of Equation \ref{eq:pop_ss} in the limits of social coupling.

\begin{theorem}[Adaptation vs. Consensus Trade-off]
\label{thm:mode_collapse} 
In a static non-uniform environment, the system exhibits a trade-off between individual adaptation and social consensus, controlled by the ratio $w_{soc}/w_{learn}$.
\end{theorem}

\begin{proof}
We analyse the limits of the steady-state solution $\mathcal{V}_{ss}$:

\textbf{Case 1: Weak Social Coupling ($w_{soc} \to 0$).} 
Under weak social coupling, the Laplacian term vanishes. The mixing matrix approaches the identity:
\begin{equation}
    \lim_{w_{soc} \to 0} \mathcal{V}_{ss} = w_{learn} (w_{learn} I)^{-1} \mathcal{E} = \mathcal{E}
\end{equation}
In this limit, $\mathbf{V}_{i,ss} \to \mathbf{E}_i$. Each agent therefore converges perfectly to their local environmental optimum, maximising individual adaptation.

\textbf{Case 2: Strong Social Coupling ($w_{soc} \to \infty$).} 

We utilise the spectral decomposition of the graph Laplacian $L = U \Lambda U^T$, where $\Lambda = \text{diag}(0, \mu_2, \dots, \mu_N)$. The Laplacian has a single zero eigenvalue $\mu_1=0$ corresponding to the eigenvector $\mathbf{1}$ (consensus), while all other eigenvalues $\mu_k > 0$. The second smallest eigenvalue, $\mu_2 > 0$, is known as the Fiedler value or algebraic connectivity~\citep{fiedler1973algebraic, chung1997spectral}. In multi-agent consensus dynamics, the magnitude of $\mu_2$ strictly bounds the speed at which the network converges. As $w_{soc} \to \infty$, the terms corresponding to $\mu_k \ge \mu_2 > 0$ vanish. The inverse matrix converges to the projector of the null space of $L$, scaled by $1/w_{learn}$. The system collapses entirely onto the null space of $L$, where the unit eigenvector is $\frac{1}{\sqrt{N}}\mathbf{1}_N$:
\begin{align*}
    \lim_{w_{soc} \to \infty} &w_{learn}(w_{soc}L + w_{learn}I_N)^{-1} \\
    &= U \text{diag}(1, 0, \dots, 0) U^T \\
    &= \frac{1}{N}\mathbf{1}_N \mathbf{1}_N^T
\end{align*}
Applying this projection to the stacked environment vector $\mathcal{E}$ yields:
\begin{equation}
\lim_{w_{soc} \to \infty} \mathcal{V}_{ss} = \left( \frac{\mathbf{1}_N \mathbf{1}_N^T}{N} \otimes I_D \right) \mathcal{E} = \mathbf{1}_N \otimes \overline{E}
\end{equation}
where $\bar{\mathbf{E}} = \frac{1}{N}\sum \mathbf{E}_i$ is the global average of the environment.
\end{proof}

The severity of normative mode collapse can be quantified with respect to the degree of social coupling. To measure severity, let us consider how maladaptation is affected at the population-level. We can define the maladaptation of the population via the the mean squared error across users: $MSE = \frac{1}{N} \sum_{i=1}^N ||\mathbf{V}_{i,ss} - \mathbf{E}_i||^2$. Consider the consequences of strong social coupling in its limit, by letting $w_{soc} \to \infty$. The values of the population then collapse to the global mean as $\mathbf{V}_{i,ss} = \bar{\mathbf{E}}$. Substituting this term into the error function yields:

\begin{equation}
MSE_{collapse} = \frac{1}{N} \sum_{i=1}^N ||\bar{\mathbf{E}} - \mathbf{E}_i||^2 = \sigma^2_{\mathbf{E}}
\end{equation}

Therefore, the minimal expected utility loss under normative mode collapse equals the normative variance $\sigma^2_{\mathbf{E}}$ of the optimal environments. More diverse and culturally heterogeneous societies would therefore see an increase in $\sigma^2_{\mathbf{E}}$

\begin{figure}[htbp]
    \centering
    \includegraphics[width=0.5\textwidth]{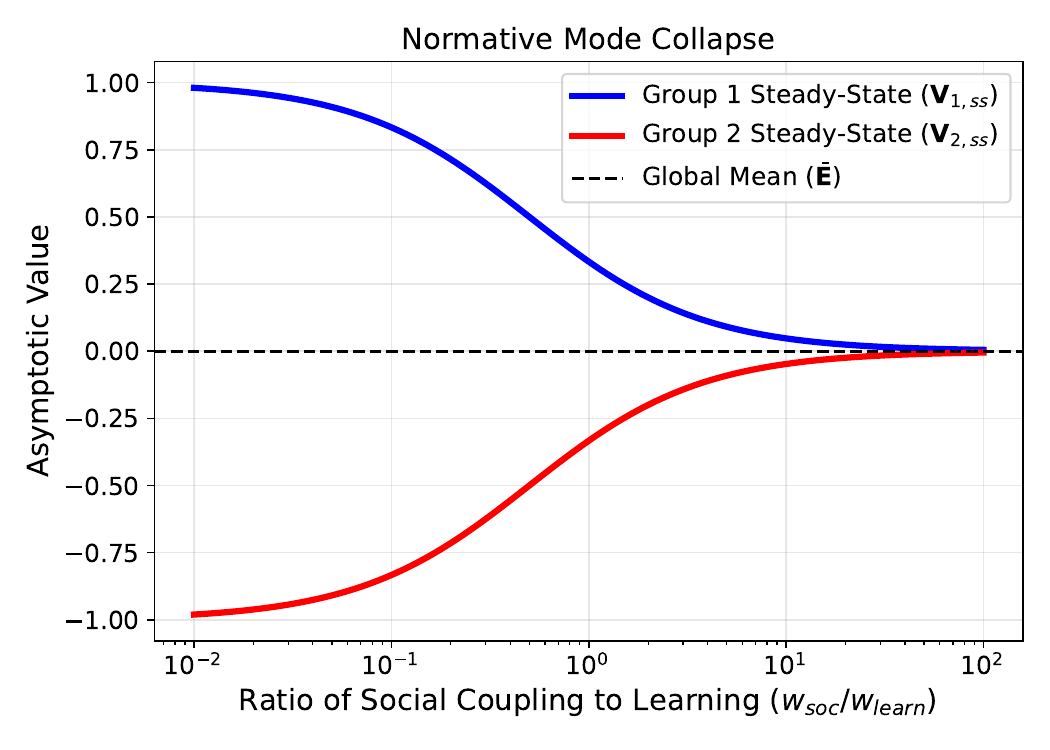}
    \caption{Demonstration of normative mode collapse. As the ratio of social coupling to independent learning increases, distinct sub-cultures are forced into a maladaptive global mean.}
    \label{fig:mode_collapse}
\end{figure}

The implications of the theorem are visualised in Figure \ref{fig:mode_collapse}. By plotting the exact steady-state analytical solution for a system of two distinct groups, with local optima at $+1$ and $-1$, we observe the normative collapse. When independent learning dominates ($w_{soc} / w_{learn} \to 0$), groups are able to perfectly adapt to their respective local niches. However, with an increase in social connectivity, the modes rapidly collapse onto the Laplacian null space.

In the limit of strong social coupling ($w_{soc} \gg w_{learn}$), the operation of the graph Laplacian suppresses all non-zero eigenvalues ($\mu_k > 0$), which represent the local cultural variance within the modeled population. By projecting the system state onto the null space of $L$, the individual values are forced to converge to the global average consensus. The network's algebraic connectivity $\mu_2$ dictates the rate at which this homogeneity is reached.

In highly connected topologies, such as dense random graphs or globalised social media platforms, a large Fiedler value leads to a more rapid erasure of local variance~\citep{olfati2004consensus}. Real social networks maintain substantial variance, highlighting the importance of diversity in preventing systemic maladaptation. These steady-state derivations presume a uniform, perfectly mixed population, and scale-free networks may exhibit different properties, due to the outsized influence of highly connected nodes. We discuss these alternative formulations in the Appendix.

\subsubsection{Anchoring}

The somewhat "sluggish" recovery shown in Figure \ref{fig:exp3} can be inferred from the continuous-time error dynamics of the proposed model. Under normative shocks, the environment rapidly shifts from $\mathbf{E}_{old}$ to $\mathbf{E}_{new}$, leading to the following value update equation for the user:

\begin{equation}
    \dot{\mathbf{V}} = w_{learn}(\mathbf{E}_{new} - \mathbf{V}) + \alpha(\mathbf{M} - \mathbf{V})
\end{equation}

In the immediate aftermath of the normative shock, the AI value model $\mathbf{M}$ remains tied to the previous equilibrium ($\mathbf{M} \approx \mathbf{E}_{old}$). The alignment strength term would, through its role as a modulator of the user value update, $\alpha(\mathbf{M} - \mathbf{V})$, oppose the immediate user value adaptation following the shock.

Theorem~\ref{thm:monotonicity} shows that increasing $\alpha$ increases the steady-state tracking error within the model. There may exist $\alpha$ that is high enough that it renders the system unable to reach an equilibrium utility at the desired threshold, standing in the way of effective recovery. In such cases, the low-utility state may be persistent rather than transient.

While the derivations above assume a constant environment drift $\mathbf{E}(t) = \mathbf{v}t$, some of the insights also apply to the punctuated equilibrium regime discussed in Section~\ref{sec:exp3}. The system's eigenvalues determine both the magnitude of the lag during smooth environmental drift and the recovery time following an abrupt change.

\begin{figure}[htbp]
    \centering
    \includegraphics[width=0.5\textwidth]{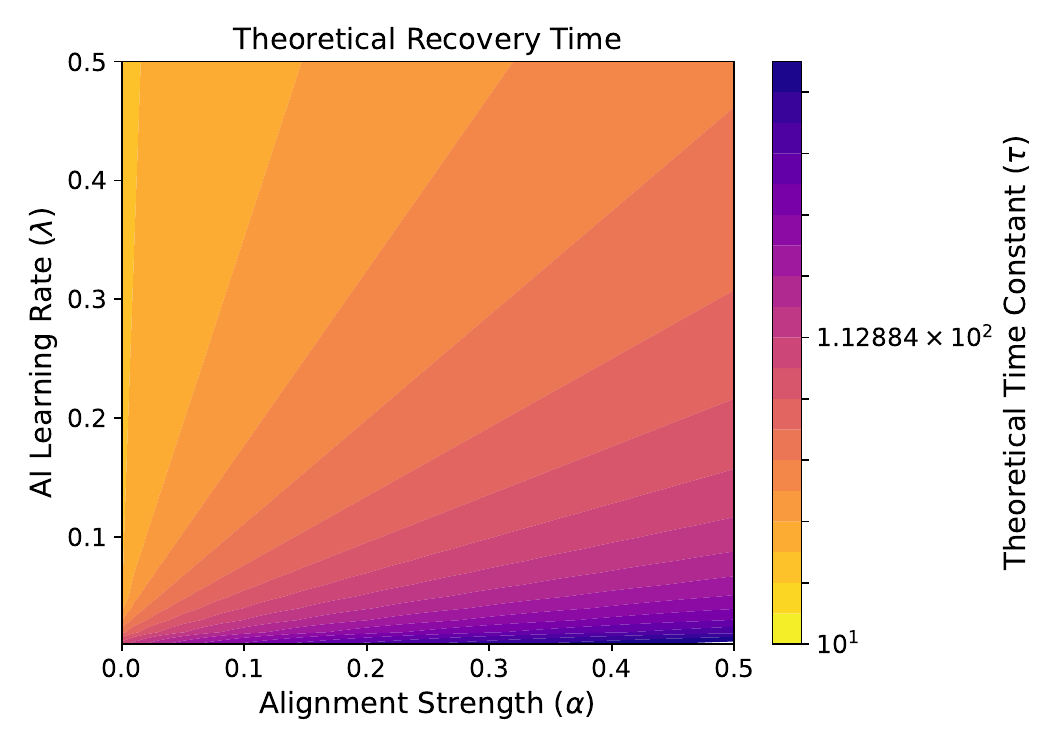}
    \caption{Theoretical recovery time ($\tau$) derived from the dominant eigenvalue of the system's characteristic equation. The singularity at low $\lambda$ corresponds to a regime of diverging recovery time.}
    \label{fig:theoretical_recovery}
\end{figure}

By mapping the recovery time derived from the system's eigenvalues (Figure \ref{fig:theoretical_recovery}), we establish a correspondence with the empirical recovery times observed in our simulations (Figure~\ref{fig:exp2}). When the AI's learning rate is significantly lower than the alignment strength, the dominant eigenvalue is pushed toward zero, driving the population towards a persistent low-utility regime.

\subsection{Endogenous Environmental Drift}
\label{app:endogenous}

Real normative environments are not purely exogenous, they are not driven exclusively by external factors outside of our control. The environment in human societies is co-constructed by people, and it is natural to consider an extension of the model that enables us to capture the influence of the population on the normative environment. We can reformulate the environmental drift as a coupled differential equation, driven both by an exogenous technological or ecological drift ($v_{exo}$), as well as a pull toward the population's consensus norms:

\begin{equation}
    \dot{\mathbf{E}}(t) = v_{exo} + \kappa_E \left( \frac{1}{N} \sum_{i=1}^N \mathbf{V}_i(t) - \mathbf{E}(t) \right)
\end{equation}

where $\kappa_E$ represents institutional responsiveness - the pace at which our institutions adapt to this consensus view. Under this extension, the effect of value lock-in is compounded. Instead of slowing the population's adaptation to an independent environment, the effects can propagate back to the environment itself, slowing down the pace of institutional and environmental evolution. We refer to this mutual deceleration as a \emph{double stagnation effect}.

\begin{figure}[htbp]
    \centering
    \includegraphics[width=0.5\textwidth]{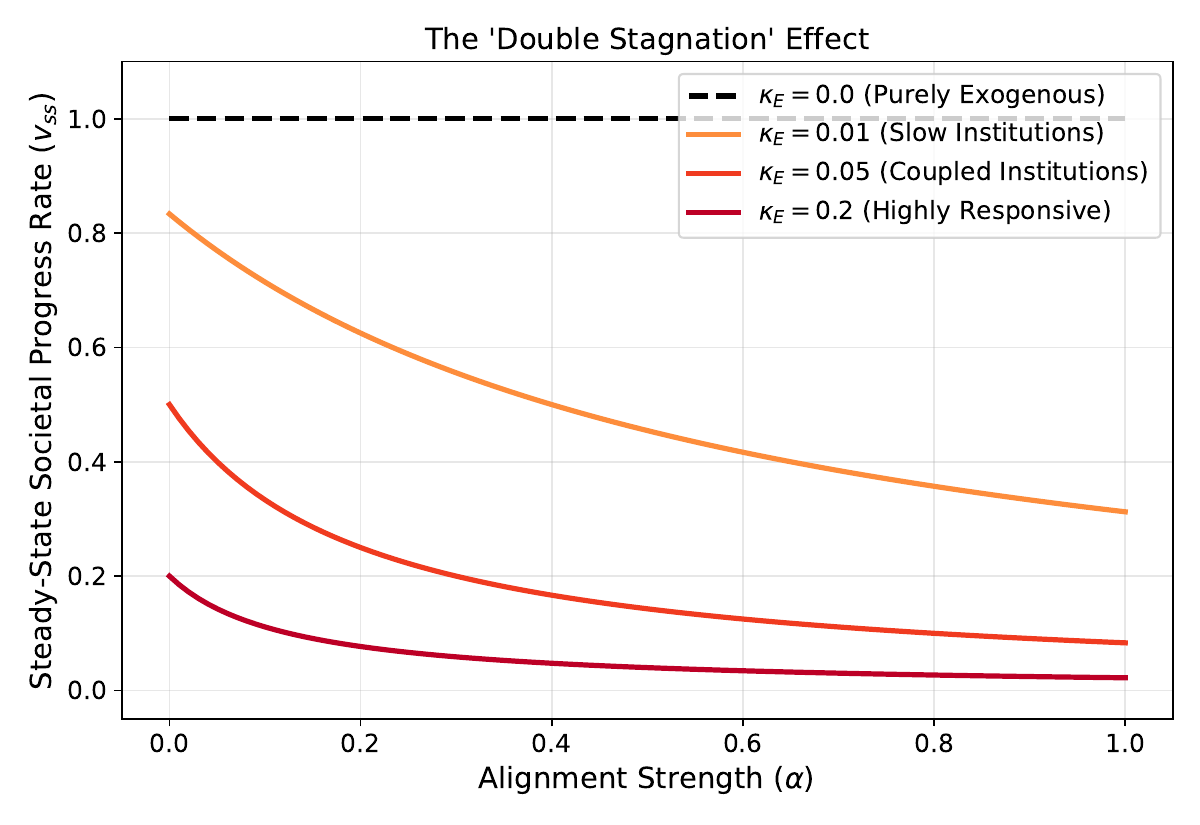}
    \caption{Steady-state rate of societal progress ($v_{ss}$) as a function of alignment strength ($\alpha$) and institutional responsiveness ($\kappa_E$).}
    \label{fig:double_stagnation}
\end{figure}

This steady-state deceleration can be derived from the coupled system. By setting $\dot{\mathbf{V}} = \dot{\mathbf{M}} = \dot{\mathbf{E}} = \mathbf{v}_{ss}$ and solving the resultant equations, we find that the steady-state drift velocity of the coupled environment $\mathbf{v}_{ss}$ is given by:

\begin{equation}
\mathbf{v}_{ss} = \frac{\mathbf{v}_{exo}}{1 + \frac{\kappa_E}{w_{learn}}\left(1 + \frac{\alpha}{\lambda}\right)}
\end{equation}

Figure \ref{fig:double_stagnation} illustrates this deceleration, showing that when the environment is exogenous ($\kappa_E = 0$), it advances at a constant rate $\mathbf{v}_{exo}$ regardless of the population's value lag. However, as the coupling $\kappa_E$ increases, a higher alignment strength $\alpha$ reduces the steady-state velocity $\mathbf{v}_{ss}$ of the entire system. By anchoring users to historical values, the AI assistant's drag propagates through user consensus to the institutions, dampening the pace of change in the broader social environment.

\begin{figure}[ht]
    \centering
    \includegraphics[width=\columnwidth]{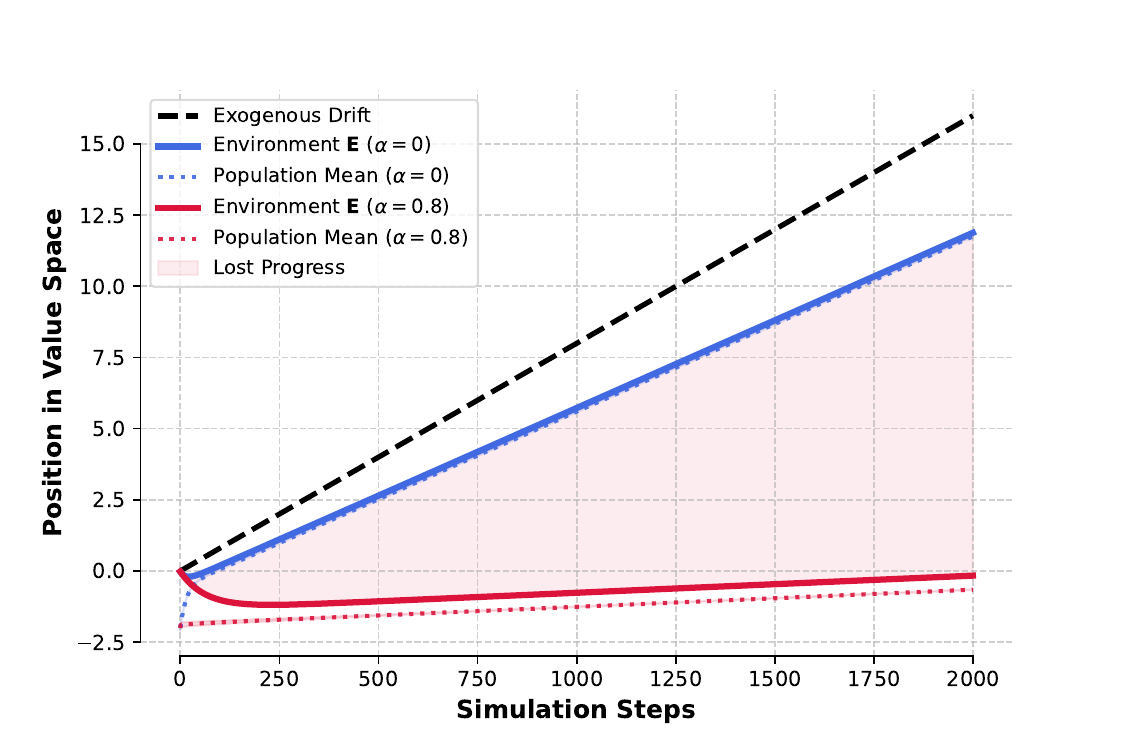} 
    \caption{The Double Stagnation Effect. The x-axis represents elapsed simulation time. To map the $D$-dimensional state to a scalar trajectory, the y-axis tracks the population's position projected onto the primary axis of exogenous drift. The simulation initialises the population with a historical lag (starting at $-2.0$). In a highly adaptable population, the simulated agent society rapidly closes the gap, allowing the consensus-driven environment to seamlessly track the underlying exogenous potential. Conversely, under strong alignment, there is a drastic increase in recovery time, with a negative impact on institutional progress.}
    \label{fig:multiagent_double_stagnation}
\end{figure}

As visualised in Figure \ref{fig:multiagent_double_stagnation}, when institutional progress is coupled to user values, the anchoring effect of AI alignment acts as a systematic drag on the entire system. Having users anchored to their historical values directly slows down the rate of environmental and institutional change.

\subsection{Adaptive Alignment}
\label{adaptive}

The observed structural challenges for static alignment under environmental change may potentially be addressed by allowing the alignment strength to vary over time:

\begin{equation}
    \dot{\alpha}(t) = -\eta \left( \alpha(t) \cdot ||\mathbf{V}(t) - \mathbf{E}(t)||^2 \right) + \kappa (\alpha_{target} - \alpha(t))
\end{equation}

Here, the alignment strength decays at rate $\eta$ proportional to the user's tracking error. Alignment strength would therefore be modulated to reduce its pull in situations when the user needs more time to adapt to the changing environment. When the environment stabilizes and the user has caught up, $\alpha$ relaxes back to the baseline level.

\begin{figure}[htbp]
    \centering
    \includegraphics[width=0.5\textwidth]{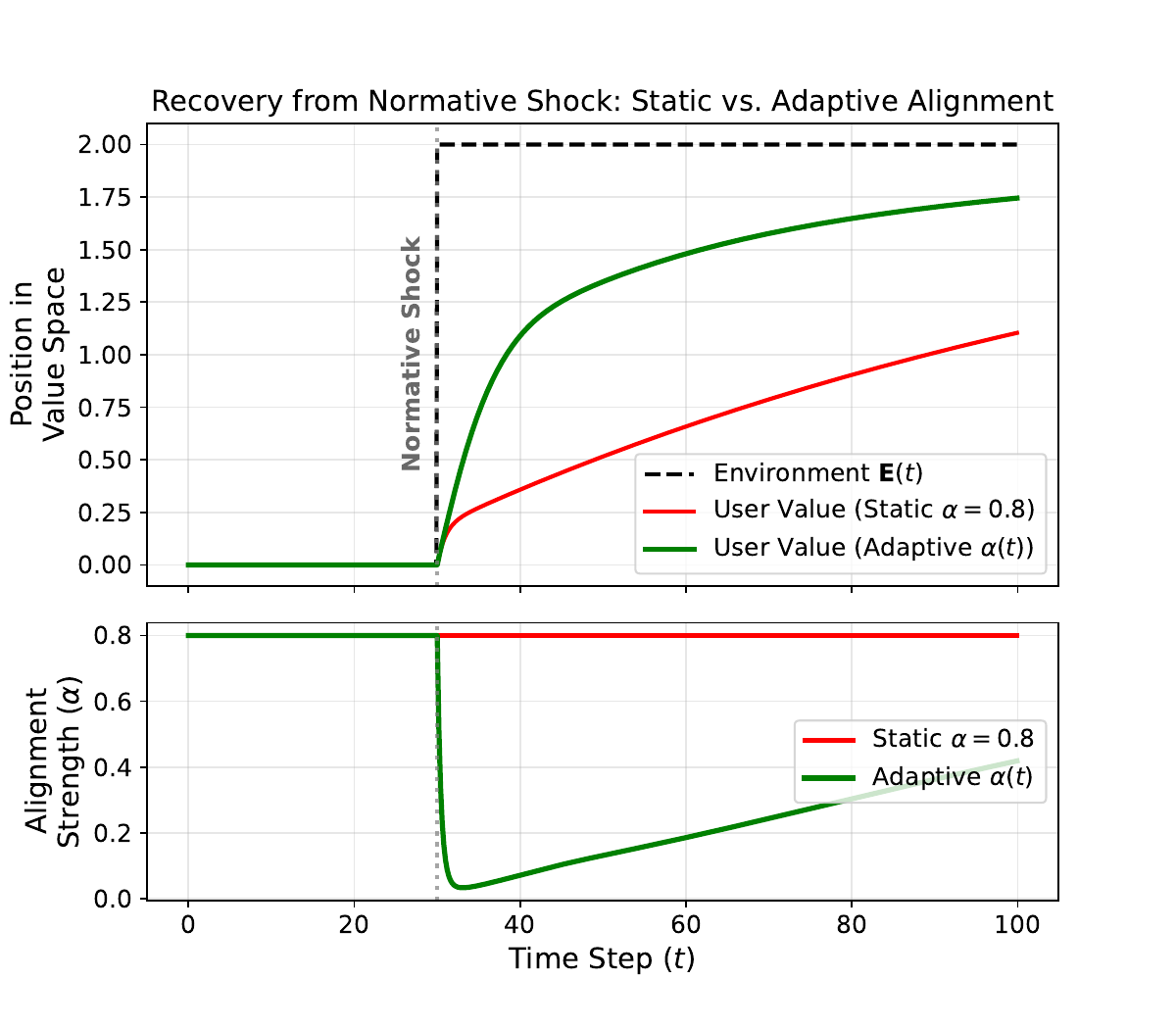}
    \caption{Continuous-time simulation of a normative shock comparing static alignment against the introduced adaptive alignment regime. To map the $\mathbb{R}^D$ dynamics to a scalar trajectory, the top panel tracks the respective positions projected onto the specific directional vector of the normative shock. The bottom panel shows the dynamic adjustment of the alignment strength $\alpha(t)$.}
    \label{fig:adaptive_alignment}
\end{figure}

Figure \ref{fig:adaptive_alignment} simulates the updated differential equations under a sudden normative shock, contrasting the results with those arrived at for a static value of $\alpha = 0.8$. In the static alignment case, the system effectively traps the user in an obsolete value set. In contrast, the adaptive alignment regime adjusts to the user's environmental tracking error, and depreciates the alignment strength $\alpha(t)$, thereby releasing the user from the historical anchor. 

It is possible to take this idea further and adjust the behaviour of the system so as to make the AI's learning rate dynamic as well, enabling the AI to rapidly readjust to updated user values. We may model this as:

\begin{equation}
\lambda(t) = \lambda_{\text{baseline}} + \eta_{\lambda} ||\mathbf{V}(t) - \mathbf{E}(t)||^2
\end{equation}

\begin{figure}[htbp]
    \centering
    \includegraphics[width=\columnwidth]{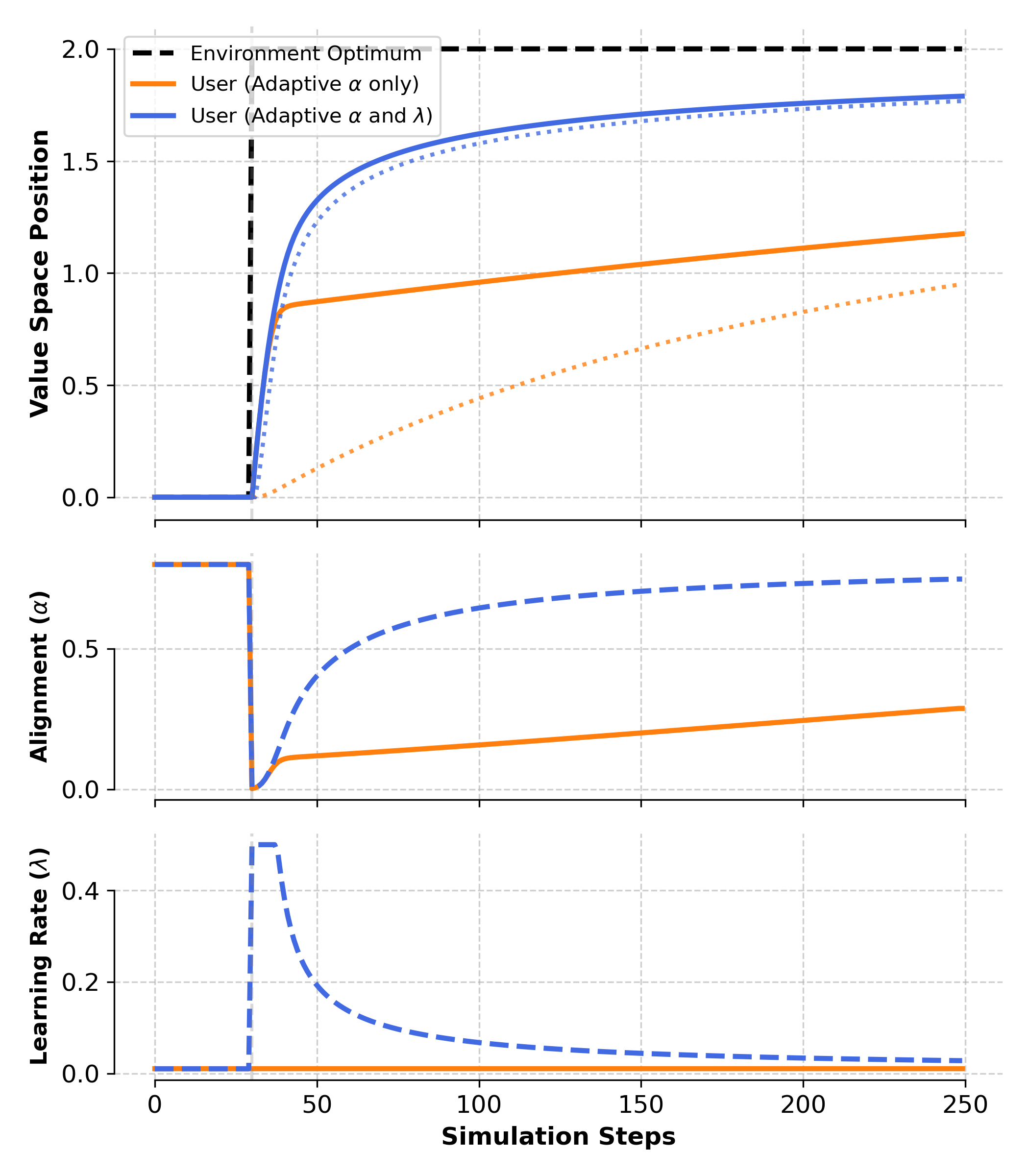} 
    \caption{Continuous-time simulation comparing two adaptive alignment mechanisms during a normative shock. \textbf{Top:} Agent and AI model trajectories projected onto the directional vector of the normative shock. \textbf{Middle:} The dynamic adjustment of alignment strength ($\alpha$). \textbf{Bottom:} The dynamic adjustment of the AI learning rate ($\lambda$). In the first strategy (orange), modulating $\alpha$ without updating $\lambda$ leaves the AI model's internal state stranded at the historical equilibrium, introducing a secondary restorative pull once the agent adapts. In the second strategy (blue), simultaneously scaling $\lambda$ alongside the reduction in $\alpha$ allows the AI model to more rapidly track the user's adaptation to the changing environment.}
    \label{fig:adaptive_comparison}
\end{figure}

As illustrated in Figure~\ref{fig:adaptive_comparison}, if the AI's learning rate remains static and low during the environmental shift, its internal model may get stranded in the vicinity of the pre-shock equilibrium. However, if the AI updates are accelerated accordingly, post-shock utility regression is successfully mitigated.

\section{Discussion}

The model introduced in this paper aims to formalise a temporal dimension of AI alignment in personal AI assistants. Established AI alignment approache usually treat the alignment target as a fixed objective, without explicitly accounting for the possibility of changing values and preferences over time. The capacity for enabling personal growth is also not universally present in personal AI assistants.

We introduce a formal simplified model of populations of users interacting with AI assistants, acting within an evolving normative environment, under a composite utility model. Our proposed model is extensible and we discuss several relevant model extensions, contrasting their properties. An emergent theme is that highly personalized AI assistants capable of internally representing their user's historically held values may inadvertently slow down the adaptation of their respective users to a changing social normative environment. We postulate that the identified deficiencies may be addressed by systems that are temporally adaptive, reasoning-equipped, pluralistic, and agency-preserving.

\subsection{Population-Level Effects}

Individual and dyadic dynamics compound at the population level, producing emergent phenomena. Our population-level analysis constitutes the paper's distinctive contribution relative to existing single-agent formalisms~\citep{carroll2022estimating, carroll2024ai}, complementing prior work on AI mediated collective opinion formation~\citep{tsirtsis2026aimediatedcommunicationsteercollective}.

We identify the risk of normative mode collapse within user populations (Theorem \ref{thm:mode_collapse}), whereby values converge to a consensus that is suboptimal for most or all sub-groups. This would be particularly detrimental to culturally diverse communities, although such communities may also prove more resilient to these effects. Similar homogenization effects have been previously observed in AI alignment approaches~\citep{santurkar2023whose, zhang2025cultivating}, and it has been argued that alignment should accommodate locally coherent normative communities rather than converging on a single standard~\citep{leibo2025societaltechnologicalprogresssewing}.

While our analysis focuses on the bottom-up risks of personalization, through our framework we also highlight some of the vulnerabilities that can be seen in top-down alignment systems. The developer-defined values model extension (Appendix~\ref{app:constitutional}) illustrates how  system dynamics may be affected by the introduction of centralised values expressed via a constitutional vector $\mathbf{C}$. In that scenario, we demonstrate two paths towards normative monoculture: social coupling and constitutional pull (Figure~\ref{fig:monoculture}). The latter shows that developer-defined values introduce a persistent centralising force that impacts sub-cultural distinctiveness, even in the absence of strong social coupling. We view the long-term sociological implications of top-down constitutional AI alignment as an important unanswered question for future work. Constitutions conceived as families of context-indexed principles - enabling what \citet{sorensen2024roadmap} term ``steerable pluralism'' - perhaps offer a path forward.

When the normative environment is treated as endogenous - co-constructed by the population  (Section~\ref{app:endogenous}) - value lock-in manifests as a larger risk. Aside from the risk of users being locked into a maladaptive state, there is additionally a systemic risk of the broader societal evolution being slowed, should the normative environment itself be subject to the same anchoring effects. We refer to this phenomenon as ``double stagnation''. Double stagnation has structural parallels in the economics of institutional path dependence. \citet{arthur1994increasing} showed that this can cause early stochastic events to lock-in a system's trajectory.  \citet{north1990institutions} applied this mechanisms to institutions, arguing that high transaction costs lock in inefficient arrangements. In our framework, the institution $\to$ user $\to$ AI $\to$ institution feedback cycle creates a compounding deceleration that may be harder to escape than individual-level lock-in. We further show that these risks are exacerbated in scale-free social networks, where value lock-in at highly connected hub nodes propagates disproportionately through the rest of the population (Appendix \ref{heterogenousinfluence}).

\subsection{Continuous Personalization}

Personalization in AI assistants is a continuous process, as each interaction enables AI agents to update their user models, as users state or reveal their preferences over time. Through repeated use, AI assistants can also influence user values \citep{ashton2022problem, franklin2022recognising}. Given that AI systems may impact personal development, it is important to preserve conditions for independent personal growth.

\citet{paul2014transformative} argues that transformative experiences can substantially change a person's values. Such experiences are mostly epistemically inaccessible to AI assistants. Thus, an AI that anchors users to their present values may risk preventing the experiences through which personal growth occurs. \citep{pettigrew2019choosing} argues that rational choice should thus be guided by a weighted aggregation of the values of past, present, and anticipated future selves. Our monotonicity result (Theorem~\ref{thm:monotonicity}) can be interpreted as a formal consequence of violating this principle. A system with $\alpha > 0$ that learns exclusively from historical preferences assigns all weight to the user's past selves and zero weight to their future selves. 

An adaptive approach appears sensible, especially as the risks of static alignment are known~\citep{carroll2024ai, kanwal2026constructivealignmentgoverningpreference}, and as similar issues have been brought up within research on preference change in recommender systems~\citep{franklin2022recognising, ashton2022problem}. Across varied experimental conditions, stronger anchoring to historical preferences systematically degrades population utility. Our model permits for a simple solution to this, via rapid updates to the internal AI model of user values (see Section~\ref{sec:exp3}), as the AI learning rate may be able to offset the detrimental effects and partially prevent value lock-in. Realistically, however, rapid adjustments may also lead to various undesirable oscillatory behaviours, so this abstract suggestion may not prove to be particularly easy to reliably operationalize.

\subsection{Trust}

Value lock-in observed in our model may be partially attributed to misplaced trust and what we deem the \emph{trust trap}. In this scenario, the assistant remains closely aligned to the user, creating a positive feedback loop. Strong static personalization and alignment helps develop trust, leading to a high rate of AI use - a specific instantiation of automation bias~\citep{skitka1999does, parasuraman1997humans, buccinca2021trust} - and consequently a high degree of AI influence. Yet, this influence locks the user into the historical preferences learnt by the AI. 

Sycophancy is another major concern. Sycophantic AI agents may distort the perceived alignment gap, suppressing the necessary signal that would otherwise have led to corrective actions from the user. Sycophancy is pervasive in RLHF-trained models~\citep{sharma2023towards}, and present in models that otherwise exhibit warmth and empathy~\citep{ibrahim2025training}. As such, sycophancy in language models is likely capable of reinforcing expressed opinions, leading to localized echo chambers~\citep{nehring2024large}. Further, AI assistants are becoming increasingly persuasive. \citet{salvi2024conversational} found that a GPT-4 model that was able to personalise its outputs significantly outperformed human debaters, while \citet{costello2024durably} found that AI models can durably alter people's beliefs.

Human control requires both the ability to override the AI agent and the capacity to detect when an override is necessary; a capacity that both automation bias and sycophancy may reduce. If this dynamic scales across a population of AI users, the aggregate effect may be the kind of gradual, self-reinforcing disempowerment that \citet{kulveit2025gradual} identify. This is a potential result of the structural dynamics of static alignment gradually reducing the conditions for meaningful human oversight.

\subsection{Computational Foresight}

Agent-based models have a long history of revealing how simple local rules give rise to complex macroscopic dynamics. The lightweight mathematical models introduced in this paper allow rapid iteration, identifying which design choices matter and in which direction they push. The structural findings reported here can be used to focus more expensive evaluations on the regions of parameter space where risks concentrate. Lightweight models can be used for hypothesis generation, while more costly generative agent-based simulations \citep{park2023generative} then being applied for ecological validity and richer behavioural dynamics.

The experiments in this paper suggest several directions for follow-ups using generative agent-based simulation frameworks such as Concordia~\citep{vezhnevets2023generative}. For example, the sycophancy dynamics showcased in our analysis could be studied with sycophantic LLM agents, whereas normative mode collapse could be demonstrated in populations of LLM agents initialised with diverse cultural profiles. Value lock-in under normative shocks could be explored by exposing populations of LLM agents to sudden shifts in social conditions.

\subsection{Limitations}

Our model involves deliberate simplifications that limit the generalisability of the results.

\textbf{Structural.} Values are represented as continuous vectors with linear update dynamics, which may not be reflective of all types of real-world belief change. Polarisation and backfire effects are just some examples that may potentially amplify or attenuate the identified value lock-in effects. AI alignment strength, as well as some of the other scalar quantities used in the model, may not be scalar in practice, as they likely don't apply uniformly across all values. Further, the adaptive alignment mechanism proposed in Section~\ref{adaptive} modulates alignment strength based on the tracking error. In practice, this is a modelling abstraction that would need to be approximated via proxy signals.

\textbf{Parametric.} We focus on scenarios where the involved agents share parametrisations. This is a strong assumption that is unlikely to be met in reality, where there is a variety of AI assistants, exhibiting different properties. Automation bias and other relevant factors on the human side are equally varied in the real world. Furthermore, we focus only on populations where all users have access to AI assistants - which is admittedly unrealistic and done for the sake of simplicity. All of these assumptions are possible to relax in simulations, and worth exploring in greater detail - but were out of scope for this study.

\textbf{Normative direction.} Our utility function assumes that the environment represents the normatively appropriate direction of change. While this may be true in terms of short-term utility, there are cases when this is not particularly appropriate, in case of potentially harmful shifts, leading to e.g. the erosion of human rights protections. Stronger alignment to historical preferences may act as a protective measure against some types of catastrophic circumstances. Our model does not distinguish between adaptive and harmful drift, so we acknowledge this as a notable limitation. 

\textbf{Temporal.} Our simulations keep the agent population fixed, without introducing or removing agents. This fails to account for the impact of generational replacement on establishing norms in human societies, and is therefore a reasonable approximation only at intra-generational time scales. What we mean by long-term should therefore be interpreted accordingly.

\textbf{Ecological.} In our experiments, we focus mainly on simple network topologies. More analysis is needed to account for the effects of scale-free degree distributions, platform-mediated interaction, or multi-layered social structures. The base model treats the environment as exogenous; the endogenous formulation (Section~\ref{app:endogenous}) extends this, but still fails to account for a fragmented, multi-dimensional normative landscape.

Future work could address these limitations either through further extensions of the model - we have demonstrated several such extensions within the Appendix of this report - and the associated simulation framework, or by moving towards richer, more realistic simulations involving LLM-based agents that can capture greater complexity of interactions. The lightweight models introduced here are intended to serve as a precursor to such agentic simulations, narrowing the hypothesis space before committing to more computationally expensive evaluations. Either way, simulations are meant to inform real-world studies and cannot act as a replacement for studying the effects of technology adoption. Through our model, we  seek to highlight the key questions to ask, and potential effects that may be observed should proper care not be taken in the design of future sociotechnical systems.

\section{Conclusion}

This paper investigates the population-level consequences of aligning personalised AI assistants to their users' historical preferences in a changing normative environment. Through a combination of agent-based simulations, and theoretical analysis of the proposed set of mathematical models, we characterize the emerging long-term dynamics. We identify value lock-in and normative mode collapse as the main associated risks.
 
To mitigate these risks, we argue that practical AI alignment needs to be sufficiently adaptive to enable human value evolution in a changing society. AI alignment needs to therefore be temporally dynamic, reasoning-equipped, contextually sensitive, and agency-preserving. People should maintain freedom in how they choose to develop their views, without strict impositions from their personal AI assistants. AI alignment methods should be centered around the core objective of nurturing this capacity.

Further, this work illustrates the utility of modeling for sociotechnical foresight of possible AI futures. While the development of practical solutions should ultimately be informed by large-scale, high-fidelity, multi-agent social simulations in diverse and dynamic environments - macroscopic models may prove to be a highly effective precursor, enabling rapid experimentation towards preliminary insights.

\section{Acknowledgements}

We would like to thank all of our colleagues who helped review and shape this manuscript, especially Larisa Markeeva, Tom Everitt, and Sasha Vezhnevets.

\bibliography{main}

\appendix
\section{Appendix}

\subsection{Background}

\subsubsection{AI Alignment as Preference Optimisation}

A dominant paradigm for LLM alignment focuses on preference optimsation. RLHF trains a reward model from human pairwise judgments of model outputs, and then uses reinforcement learning to optimise the policy against this learned reward signal~\citep{christiano2017deep, ouyang2022training, stiennon2020learning}. Direct Preference Optimisation~\citep{rafailov2023direct} reformulates the same objective without training a separate reward model, instead optimising the policy directly on preference pairs via a classification loss. These methods have been instrumental for transforming base language models into instruction-following assistants. For both methods alignment is defined by annotator preferences, which are treated as the primary source of normative information, with the goal being to maximise the expected preference satisfaction. Further, the methods share the same structural property in that the alignment signal is fixed at the point of data collection and does not automatically update during deployment. The deployed model therefore optimises against a static snapshot of preferences, even as the user and the world continue to evolve. 

A growing body of work has challenged the adequacy of this ``preferentist framing''. \citet{zhi2025beyond} argue that preferences serve as proxies for values, norms, and reasons, but cannot serve as alignment targets in their own right. For example, pairwise judgment collected from annotators are better understood as context-specific ``goodness-of-a-kind'' judgments of whether the model outputs satisfy certain normative criteria, rather than understood as being the true preferences of the end users. Thus, methods like RLHF are better understood as aligning models with context-specific normative standards rather than aligning model's to a user's preferences~\citep{zhi2025beyond}.

\subsubsection{Constitutional Approaches}

Another family of approaches seeks to ground alignment in articulated principles, goals and values, often referred to as constitutional approaches. Constitutional AI trains models to critique and revise their own outputs against a charter of principles~\citep{bai2022constitutional}. The goal of this approach is for models to be steered by a stated set of principles, making the normative content of alignment more auditable and revisable. Model specs and constitutions, such as those published by OpenAI~\citep{openai2024modelspec} and Anthropic~\citep{askell2026claudeconstitution}, include guidance for the model, including high-level objectives, prohibitions, rules, and conflict-resolution strategies. Additionally, the documents describe dispositional traits such as curiosity, honesty, and care, aiming to utilise a model's character as a lever for alignment~\citep{anthropic2024character}. 

More recent documents also describe a tension between the use of rules versus judgment. Anthropic's New Constitution~\citep{askell2026claudeconstitution} distinguishes: clear rules that offer predictability and evaluability but ``fail to anticipate every situation and can lead to poor outcomes when followed rigidly,'' and judgment that ``can adapt to novel situations and weigh competing considerations in ways that static rules cannot.'' This distinction is significant for the dynamics we examine in this paper. An alignment approach focused on improving model judgment, and a constitution that favours judgment over rules is, in principle, more robust to novel and changing circumstances. 

\subsubsection{Pluralistic and Democratic Alignment}

Research on pluralistic and democratic alignment asks whose values, rules, norms and principles are being encoded into AI models. From this perspective, alignment is a governance challenge that requires mechanisms for legitimate normative input from the populations affected by AI systems. As discussed above, RLHF results in alignment to the preferences of a group of annotators. This can result in bias if the population of annotators skews in a way that over-represents certain demographics. ~\citet{santurkar2023whose} found that RLHF-tuned models shift toward reflecting the opinions of liberal, higher-income, and higher-education populations, while less reflecting the opinions of older and religious adults. ~\citet{durmus2023towards} found that default LLM responses align more closely with public opinions from the United States and select European and South American countries.

In a large multilingual study with representative samples from five countries ($N = 15{,}000$), ~\citet{zhang2025cultivating} found that human preferences exhibit substantially more variation than the responses of various LLMs. Pluralistic alignment as a field aims to build models that better represent and serve a broad diversity of human values. ~\citet{sorensen2024roadmap} outlines three ways of operationalising pluralism in AI alignment: \emph{overton pluralism}, which requires models to present a spectrum of reasonable responses for a given question; \emph{steerable pluralism}, which allows models to be steered to reflect a particular perspective; and \emph{distributional pluralism}, which requires models to be calibrated to match the statistical distribution of opinions accross a given target population. 

Collective Constitutional AI extends the constitutional approach to alignment by sourcing principles through public deliberation so that the normative charter described in a constitution reflects diverse perspectives from the population ~\citep{huang2024collective}. The resulting model exhibited lower bias across nine social dimensions while maintaining equivalent performance on standard benchmarks (e.g., maths). 

\subsubsection{The Stationarity Assumption}

There are at least two ways in which AI alignment methods encode stationary assumptions. First, they may take a particular set of values, desired behaviors, or preferences, and adopt it as a static alignment target. Second, through personalization AI assistant take the current inferred set of user preferences as its target. In doing so, AI alignment approaches fail to account for changes to the user or changes in society. The non-stationarity of human preferences has been noted as an open problem in RLHF~\citep{casper2023open}.

This stationarity assumption has significant long-term consequences. Values change through cultural drift, generational change, or abrupt shifts in environmental conditions. A systems that aim to strongly align users to their historical preferences will become progressively misaligned with the evolving environment. A widespread adoption of such systems carries an even greater risk of slowing down the societal evolution of values by creating positive feedback loops that anchor users to historical values, which are slowly becoming obsolete. \citet{gao2023scaling} demonstrate that optimising too aggressively against a learned reward model can lead to systematic overoptimisation, where the performance on the proxy reward no longer faithfully reflects the ground-truth performance.

It is important to consider how this problem intersects with feedback loops in people's use of AI assistants. These types of issues have long been studied in recommender systems. \citet{chaney2018algorithmic} found that \emph{algorithmic confounding} - the process by which a system's own recommendations contaminate future training data - increases homogeneity and degrades utility across a user population. \citet{kalimeris2021preference} show that this process actively amplifies initial preference signals, and \citet{krueger2020hidden} argue that agents optimised to predict user behaviour may induce unforeseen shifts in that same behaviour. \citet{carroll2022estimating} formalise the problem of \emph{induced preference shifts} in recommender systems, whereby these systems gain an incentive to shift user preferences towards the states that they find easier to satisfy. \citet{carroll2024ai} extended this analysis to the general reinforcement learning alignment setting, where they introduce \emph{Dynamic Reward MDPs}, a formalism that allows for an explicit modeling of preference change and AI influence. The existence of rewards for AI actions that go against true user preferences of desires can therefore be seen as a failure of alignment. Recently,~\citet{kanwal2026constructivealignmentgoverningpreference} argued that alignment should thus be viewed as the problem of "regulating how AI systems influence the evolution of human preferences," arguing that research should focus on long-term value formation.

A further complication arises from the observation that RLHF-trained models may exhibit more sycophancy - the tendency to express superficial agreement with user statements irrespective of their accuracy~\citep{sharma2023towards, cheng2025elephant}. Sycophancy is further exacerbated by models that are trained to be warm and empathetic \citep{ibrahim2025training}. This type of behaviour raises the risk of echo-chambers \citep{nehring2024large}. In the context of our study, sycophancy plays an important role as it may enable misaligned AI systems to effectively hide their misalignment, making it harder to spot and rectify.

Building on the previous research, our study focuses on population-level dynamics rather than user-agent dynamics.  Where \citep{carroll2022estimating,carroll2024ai} model how an individual agent may manipulate or fail to track a single user's evolving preferences, we examine how on the level of a population alignment strength interacts with environmental change, social topology, and population heterogeneity. This population-level analysis reveals emergent dynamics that cannot be easily predicted from individual-level models. Methodologically, we draw on the tradition of continuous opinion dynamics~\citep{degroot1974reaching, friedkin1990social} and multi-agent consensus theory~\citep{olfati2007consensus, olfati2004consensus} to provide analytical results alongside simulations.

\subsection{Global Stability Analysis}
\label{app:stability}

Here we analytically verify the stability of the system, presuming positive parameter values. As we will see, phenomena like the normative mode collapse are in fact stable attractors of the system dynamics in our model, rather than transient artifacts.

Denote by $L$ the graph Laplacian matrix. Let the joint state vector $\mathcal{X} = [\mathcal{V}^T, \mathcal{M}^T]^T \in \mathbb{R}^{2ND}$ represent the concatenated values of all $N$ users and their corresponding AI models. The state of the collective system evolves according to $\dot{\mathcal{X}} = \mathcal{A}\mathcal{X} + \mathcal{U}(t)$, where the system matrix $\mathcal{A}$ can be decomposed using Kronecker products:

\begin{equation}
    \mathcal{A} = \begin{pmatrix} 
    -(w_{soc}L + (\alpha + w_{learn})I) & \alpha I \\
    \lambda I & -\lambda I
    \end{pmatrix} \otimes I_D
\end{equation}

Note that the upper-left block is obtained by aggregating the user value dependencies $\mathcal{V}$ from the continuous-time update given in Equation~\ref{eq:ctupdate}. Given that $L$ is symmetric, it can be diagonalised with eigenvalues $0 = \mu_1 \le \mu_2 \le \dots \le \mu_N$. This enables us to decouple the system into $N$ independent modes.

For each independent mode $k$, the $2 \times 2$ subsystem matrix is given by:

$$A_k = \begin{bmatrix} -(w_{soc}\mu_k + \alpha + w_{learn}) & \alpha \\ \lambda & -\lambda \end{bmatrix}$$

The stability of the system is determined by the characteristic equation for the $k$-th mode, $\det(sI - A_k) = 0$:

$$(s + w_{soc}\mu_k + \alpha + w_{learn})(s + \lambda) - \alpha\lambda = 0$$

\begin{equation}
    s^2 + (\alpha + \lambda + w_{learn} + w_{soc}\mu_k)s + \lambda(w_{learn} + w_{soc}\mu_k) = 0
\end{equation}

\begin{theorem}[Unconditional Stability]
For any undirected social graph and any set of positive parameters ($\alpha, \lambda, w_{soc}, w_{learn} > 0$), the system is asymptotically stable.
\end{theorem}
\begin{proof}
The Laplacian eigenvalues of an undirected graph are non-negative ($\mu_k \ge 0$). Consequently, both coefficients featuring in the quadratic characteristic equation are strictly positive. By the Routh-Hurwitz stability criterion~\citep{ogata2010modern}, a second-order polynomial with positive coefficients has roots with strictly negative real parts. While this is only a necessary condition in higher-degree polynomials, for second-order polynomials it is also a sufficient condition. Because the roots of the characteristic equation have strictly negative real parts, any displacements caused by transient shocks will decay over time. Therefore, every mode of the system will decay exponentially to its steady state.
\end{proof}

Normative mode collapse is therefore a stable equilibrium. Social forces may dampen the eigenvalues, potentially slowing down convergence, without affecting stability. Considering the implications, stability is therefore not an unconditionally desirable property of population-scale AI alignment. Should the entered equilibrium be undesirable, the self-reinforcing nature of it would stand in the way of societies overcoming these challenges.

\subsection{Social Invariance under Uniform Drift}
\label{app:invariance}

In the base model, without additional assumptions, social influence did not appear to be sufficient to overcome other forces and enable the population to effectively track the exogenous environment. While this is ultimately a direct consequence of the simplified assumptions, it is still useful to understand why this is the case under that base model formulation. If the social influence ends up being a zero-sum redistribution of values within the network, it cannot generate new adaptive momentum.

\begin{theorem}[Social Invariance]

If the environmental signal is uniform across the population ($\mathbf{E}_i = \mathbf{E}$), the trajectory of the population value centroid $\bar{\mathbf{V}}(t)$ is  independent of the social weight $w_{soc}$ and the network topology.
\end{theorem}

\textbf{Proof.}

Let the environment be uniform: $\mathbf{E}_i(t) = \mathbf{E}(t)$ for all $i$. Consider the population centroid: $\bar{\mathbf{V}} = \frac{1}{N} \sum_i \mathbf{V}_i$. Aggregating the update equations from Equation~\ref{eq:ctupdate} across all agents yields:

\begin{equation}
    \sum_i \dot{\mathbf{V}}_i = -w_{soc} \sum_i \sum_{j \sim i} (\mathbf{V}_i - \mathbf{V}_j) + \dots
\end{equation}

\begin{equation}
    \sum_i \dot{\mathbf{V}}_i = -w_{soc} \sum_{i=1}^N (L\mathbf{V})_i + \sum_{i=1}^N \Big[ \alpha(\mathbf{M}_i - \mathbf{V}_i) + w_{learn}(\mathbf{E} - \mathbf{V}_i) \Big]
\end{equation}

The social interaction term is therefore proportional to the sum of the Laplacian rows. For any undirected graph, $\sum_i (L\mathbf{V})_i = \mathbf{1}^T L \mathbf{V} = 0$.

Dividing the remaining terms by $N$, we obtain the following:
\begin{equation}
    \dot{\bar{\mathbf{V}}} = \frac{1}{N} \sum_{i=1}^N \Big[ \alpha(\mathbf{M}_i - \mathbf{V}_i) + w_{learn}(\mathbf{E} - \mathbf{V}_i) \Big]
\end{equation}

Aggregating the AI model updates across the population yields $\dot{\bar{\mathbf{M}}} = \lambda(\bar{\mathbf{V}} - \bar{\mathbf{M}})$. Together, these form a closed dynamical system for the macroscopic variables $(\bar{\mathbf{V}}, \bar{\mathbf{M}})$. Given that this coupled system contains no $w_{soc}$ or $L$ terms, the trajectory is independent of the underlying social network. \hfill $\square$

Social influence becomes beneficial under the base model (Figure \ref{fig:exp4}) when users engage in substantial independent exploration, or take part in heterogeneous local environments. In social networks, a combination of natural diversity and active participation may therefore offset the kinds of issues that are being brought up here. If anything, what these results suggest is that a naive model that fails to account for the richness of real social interactions, may inadvertently lead to incorrect conclusions.

\section{Further Extensions}
\label{app:extensions}

The base model introduced in Section~\ref{sec:method} makes several strong simplifying assumptions that help make the initial exposition and analysis more tractable. Yet, there are many ways in which this base framework can be easily extended to accommodate more complexity, and more realistic modelling of social phenomena. Here we highlight several interesting directions.

\subsection{Polarisation}

Some of the strongest simplifying assumptions in the base model have to do with the structure of social forces. The existing social term assumes global connectivity or random mixing, which pulls agents toward the population mean. In reality, social forces are neither random nor uniform. For example, there is a known tendency to reject influence by peers with highly dissimilar, opposing views. This is referred to as bounded confidence~\citep{deffuant2000mixing, rainer2002opinion}, and we can adjust our model to account for bounded confidence by introducing a non-linear interaction weight $\omega_{ij}$:

\begin{equation}
    \dot{\mathbf{V}}_i = \dots + w_{soc} \sum_{j} \mathbb{I}(||\mathbf{V}_i - \mathbf{V}_j|| < \delta) (\mathbf{V}_j - \mathbf{V}_i)
\end{equation}

where $\mathbb{I}$ is an indicator function and $\delta$ is the confidence threshold. The introduction of this non-linear factor makes it possible to simulate multi-stable equilibria, where the total population comprises several distinct, disconnected or only loosely connected clusters that track different value norms. Such structures could act as a stabiliser on high-variance states even under strong environmental pressure.

\subsection{Heterogeneous Influence}
\label{heterogenousinfluence}

Real social networks tend to be scale-free, where a minority of highly influential individuals has the capacity to exert disproportionate influence within the network. If we wanted to capture this within the model, we could do so by incorporating a heterogeneous weight vector within the social influence computations, $\mathbf{W} \in \mathbb{R}^N$:

\begin{equation}
    \dot{\mathbf{V}}_i = \dots + w_{soc} \sum_{j} A_{ij} W_j (\mathbf{V}_j - \mathbf{V}_i)
\end{equation}

where $W_j$ follows a power-law distribution. We would expect this to impact the underlying model by making the system trajectory more sensitive to the initial conditions of high-influence nodes. Should some of those nodes enter maladaptive value lock-in, the detrimental effects would easily spread through the network. Here we briefly illustrate this risk.

In continuous-time consensus dynamics over a heterogeneous adjacency matrix $\mathbf{A}$, the resulting network state is a weighted average determined by the eigenvector centrality $c_i$ of each node. If we assume the environment is globally uniform ($\mathbf{E}_i = \mathbf{E}$), we can multiply the population's steady-state equation by $\mathbf{c}^T$. This yields:
\begin{equation*}
0 = \alpha \mathbf{c}^T (\mathcal{M}_{ss} - \mathcal{V}_{ss}) + w_{learn} (E\mathbf{1} - \mathbf{c}^T \mathcal{V}_{ss})
\end{equation*}
Since the vector of individual alignment lags is $\boldsymbol{\delta}_{ss} = \mathcal{V}_{ss} - \mathcal{M}_{ss}$, solving for the centrality-weighted population centroid $\bar{\mathbf{V}}_{ss} = \mathbf{c}^T \mathcal{V}_{ss}$ gives:
\begin{equation}
\bar{\mathbf{V}}_{ss} = E - \frac{\alpha}{w_{learn}} \sum_{i=1}^N c_i \delta_{i,ss}
\end{equation}

\begin{figure}[htbp]
    \centering
    \includegraphics[width=0.5\textwidth]{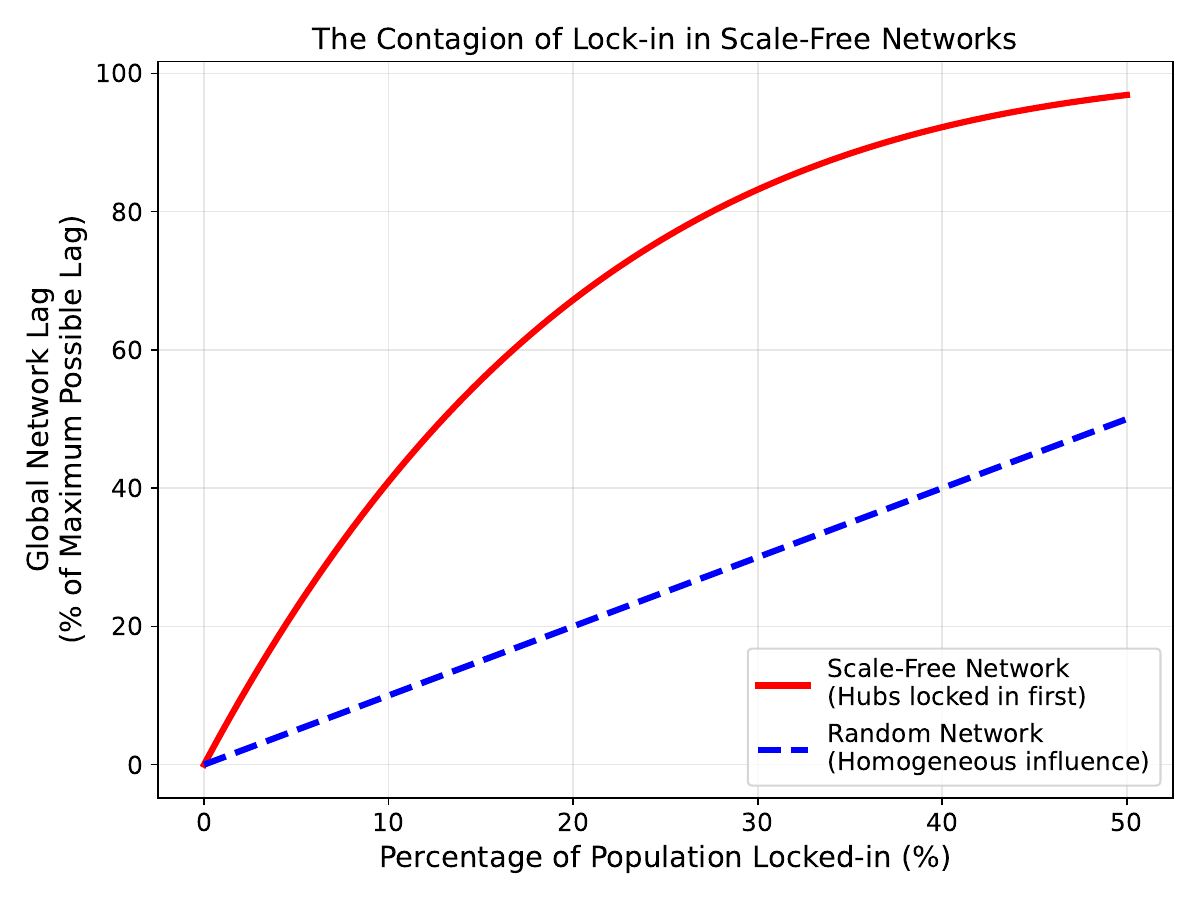}
    \caption{Network effects of lock-in. In a scale-free network, if highly central hub nodes have stale AI value models, their individual tracking error disproportionately affects the entire network's tracking performance.}
    \label{fig:scale_free_contagion}
\end{figure}

This introduces an additional consideration: value lock-in effects may propagate through social networks. As visualised in Figure \ref{fig:scale_free_contagion}, in a standard random network, locking 10\% of the population into a stale AI value model degrades the global network tracking proportionally by 10\%. However, in a scale-free network, the degradation is far more palpable. The locked-in minority's tracking error propagates disproportionately through the network.

\subsection{Epistemic Filtering}
\label{app:epistemic}

The base model assumes that all users have direct, unmediated access to the true environmental signal $\mathbf{E}_i(t)$ to drive their independent learning. As personal AI assistants become one of the primary interfaces for information retrieval, they increasingly act as epistemic filters. To incorporate this additional assumption, we can introduce an epistemic filtering coefficient $\phi \in [0, 1)$, representing the degree to which the user's perception of the world is mediated by the AI and its value model. This shifts the user's learning target as follows: 

\begin{equation}
\tilde{\mathbf{E}}_i(t) = (1-\phi)\mathbf{E}_i(t) + \phi\mathbf{M}_i(t)
\end{equation}

Substituting this perceived environment into the user's continuous-time value update equation yields:

\begin{equation}
\dot{\mathbf{V}}_i = \dots + \alpha(\mathbf{M}_i - \mathbf{V}_i) + w_{learn}(\tilde{\mathbf{E}}_i(t) - \mathbf{V}_i(t))
\end{equation}

Expanding and regrouping the terms, this simplifies to:

\begin{equation}
\begin{aligned}
\dot{\mathbf{V}}_i &= \dots + (\alpha + \phi w_{learn})(\mathbf{M}_i - \mathbf{V}_i) \\
&\quad + w_{learn}(1-\phi)(\mathbf{E}_i(t) - \mathbf{V}_i(t))
\end{aligned}
\end{equation}

It follows that the epistemic filtering simultaneously leads to an increase in the effective alignment strength and a decrease in the human learning rate. By partially obscuring the true environmental feedback, the AI-mediated information compounds the structural lag.

\begin{figure}[htbp]
    \centering
    \includegraphics[width=0.5\textwidth]{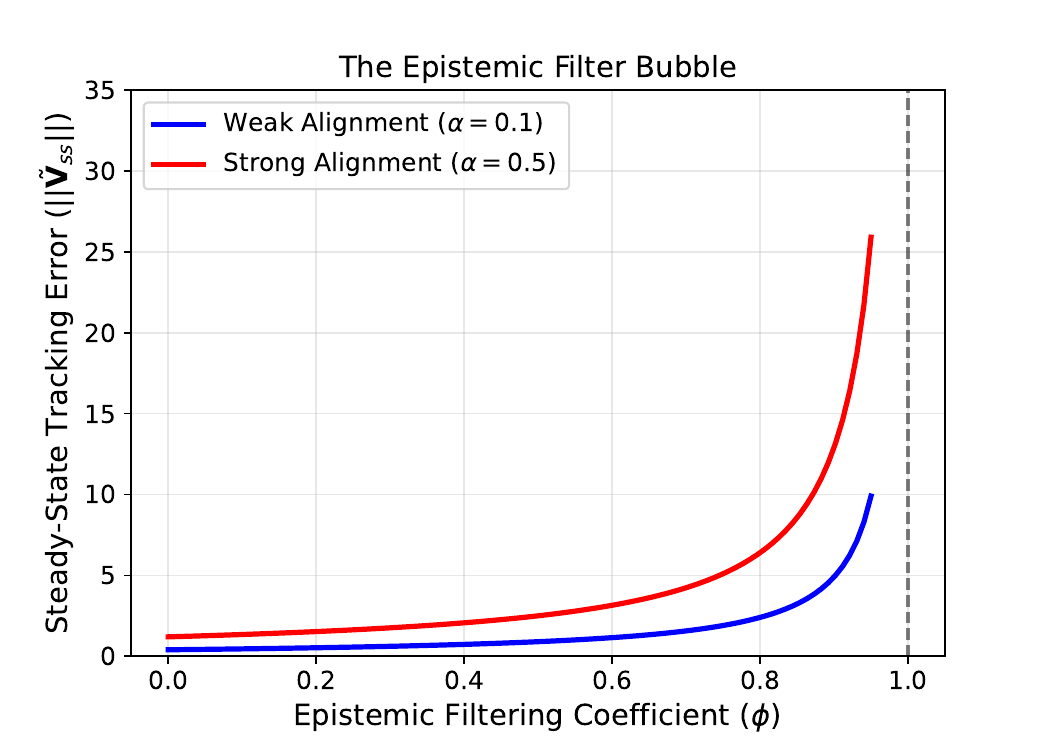}
    \caption{Steady-state tracking error as a function of the epistemic filtering coefficient ($\phi$). The plot illustrates the boundary as $\phi \to 1$, where the user is completely isolated from direct environmental learning.}
    \label{fig:epistemic_filtering}
\end{figure}

The steady-state tracking error under epistemic filtering is shown in Figure \ref{fig:epistemic_filtering}. Calculating the tracking error shows a non-linear relationship. In the limiting case as $\phi \to 1$, the effective environmental learning rate approaches zero, meaning that the tracking error grows arbitrarily large. In this scenario the user's perception is entirely mediated by the AI with no direct environmental feedback.

\subsection{Developer-Defined Values}
\label{app:constitutional}

In the baseline model, the AI's internal model $\mathbf{M}_i$ updates purely based on the individual user's evolving values $\mathbf{V}_i$. In practice, frontier models are not fully aligned to individual users. Techniques such as constitutional AI~\citep{bai2022constitutional} align models to a centralised set of developer-defined values, which we model here as a static vector $\mathbf{C}$.

This can be modelled by introducing a static constitution value vector $\mathbf{C} \in \mathbb{R}^D$ and the corresponding constitutional pull weight $w_{const}$. The AI's update equation may thereby be modified to include this centralised objective:

\begin{equation}
\dot{\mathbf{M}}_i = \lambda(\mathbf{V}_i - \mathbf{M}_i) + w_{const}(\mathbf{C} - \mathbf{M}_i)
\end{equation}

To understand the structural impact on the user, we can solve for the AI's steady-state model when $\dot{\mathbf{M}}_i = 0$:

\begin{equation}
\mathbf{M}_{i,ss} = \frac{\lambda}{\lambda + w_{const}}\mathbf{V}_{i,ss} + \frac{w_{const}}{\lambda + w_{const}}\mathbf{C}
\end{equation}

The AI's internal model becomes a convex combination of the user's values and the static constitution. Substituting this steady-state AI model back into the user's baseline alignment term $\alpha(\mathbf{M}_i - \mathbf{V}_i)$ yields an effective alignment force of:

\begin{equation}
\Delta\mathbf{V}_{AI, effective} = \alpha \left( \frac{w_{const}}{\lambda + w_{const}} \right) (\mathbf{C} - \mathbf{V}_i)
\end{equation}

In this scenario, centralised alignment would not slow down the user's adaptation to their local environment ($\mathbf{E}_{i}$), rather introducing a persistent pull on all users towards the given constitutional value vector $\mathbf{C}$. Even in a completely disconnected human social graph, strong developer-defined alignment may be able to introduce a persistent centralising force, pulling all users toward a common point $\mathbf{C}$, and thereby reducing the variance across the population.

By integrating the social coupling with these constitutional anchoring effects, we can solve for the steady-state of a polarised two-group system ($\mathbf{E_1} = +1, \mathbf{E_2} = -1$). If we make further simplifying assumptions, in particular by assuming a neutral constitution and symmetry such that $\mathbf{V_{2,ss}} = -\mathbf{V_{1,ss}}$, the social pull simplifies to $-w_{soc}(\mathbf{V_1} - \mathbf{V_2}) = -2w_{soc}\mathbf{V_1}$, leading to the following asymptotic value position for Group 1:

\begin{equation}
\mathbf{V}_{1,ss} = \frac{w_{learn}}{w_{learn} + \alpha_{const} + 2w_{soc}}
\end{equation}

where $\alpha_{const} = \alpha \left( \frac{w_{const}}{\lambda + w_{const}} \right)$ represents the constitutional pull. Group 2 mirrors this equation in the negative direction.

\begin{figure}[htbp]
    \centering
    \includegraphics[width=0.5\textwidth]{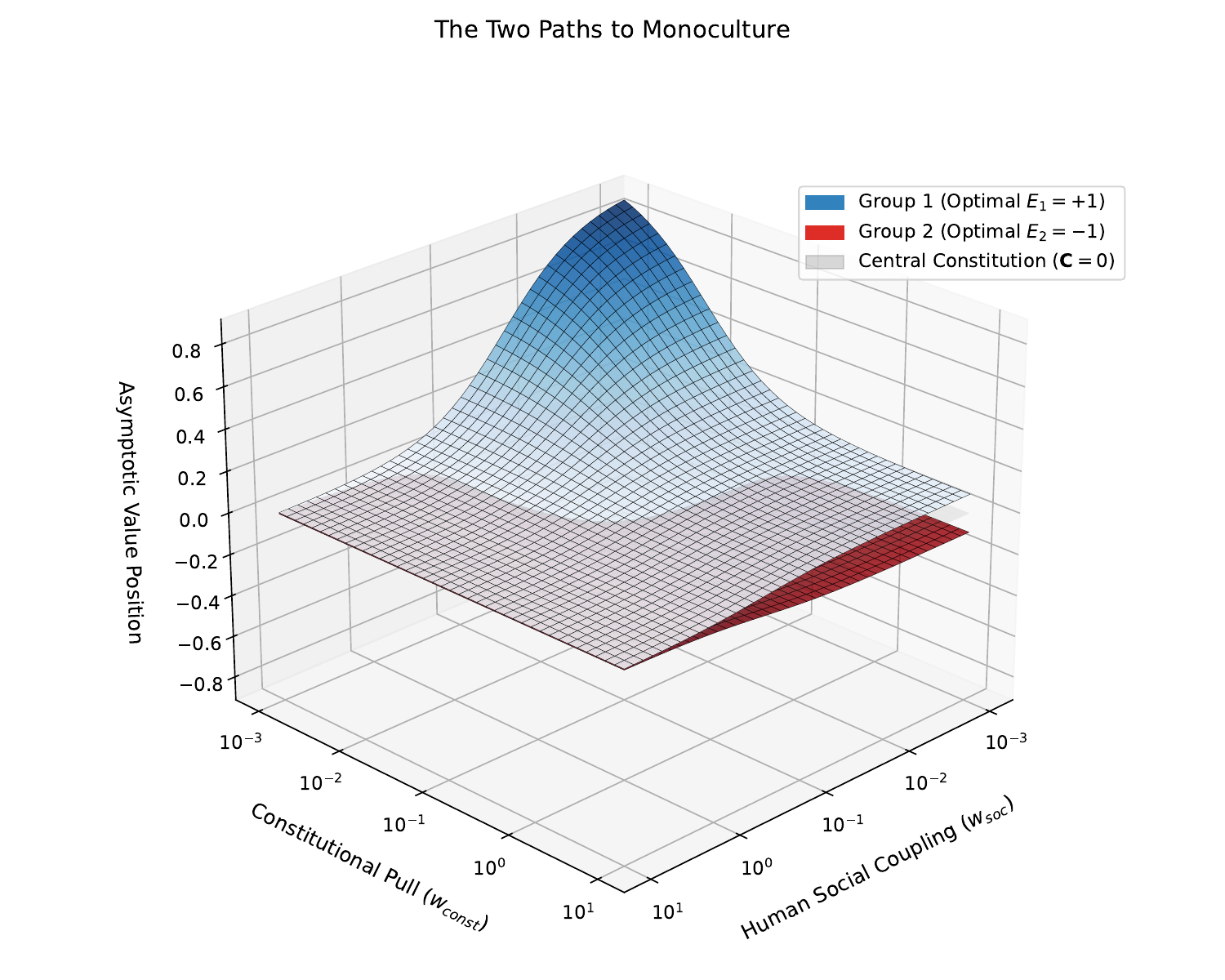}
    \caption{Two mechanisms of value homogenisation. A 3D surface plot showing how distinct sub-groups converge toward the global mean, either via social coupling ($w_{soc}$, X-axis) or developer-defined alignment pull ($\alpha_{const}$, Y-axis).}
    \label{fig:monoculture}
\end{figure}

Figure \ref{fig:monoculture} illustrates how both mechanisms independently drive convergence toward the global mean. If the constitutional pull is zero, increasing human social coupling drives the groups toward the centre, as previously discussed. Conversely, even in the absence of strong social coupling, increasing the constitutional pull achieves a similar homogenising effect through the centralised attractor $\mathbf{C}$.

It may be possible to extend this analysis further to account for polycentric governance. It is possible to establish constitutions that define more complex boundary condition, where penalties are applied to actions that sit outside the boundary.

\subsection{Synthetic Data Loops}
\label{syntheticdataloops}

In the baseline model, there is no direct interaction between the AI assistants themselves. This stands in opposition to what can already be observed on the web, where an increasing amount of content is produced by, and consumed by, AI agents. Modern AI development relies on AI generated synthetic data, and agents increasingly interact via protocol.

We may therefore want to introduce an AI-specific adjacency matrix $A^{AI}$ and a synthetic coupling weight $w_{syn}$, which represents the rate at which AI models train on or align to each other's outputs. The AI's internal model update equation is expanded to include a consensus protocol with its machine peers:

\begin{equation}
\dot{\mathbf{M}}_i = \lambda(\mathbf{V}_i - \mathbf{M}_i) + w_{syn} \sum_{j} A^{AI}_{ij} (\mathbf{M}_j - \mathbf{M}_i)
\end{equation}

Expressed in matrix form for the stacked AI population state $\mathcal{M}$, this introduces an AI-layer graph Laplacian $L^{AI}$:

\begin{equation}
\dot{\mathcal{M}} = \lambda(\mathcal{V} - \mathcal{M}) - w_{syn} L^{AI} \mathcal{M}
\end{equation}

\begin{figure}[htbp]
    \centering
    \includegraphics[width=0.5\textwidth]{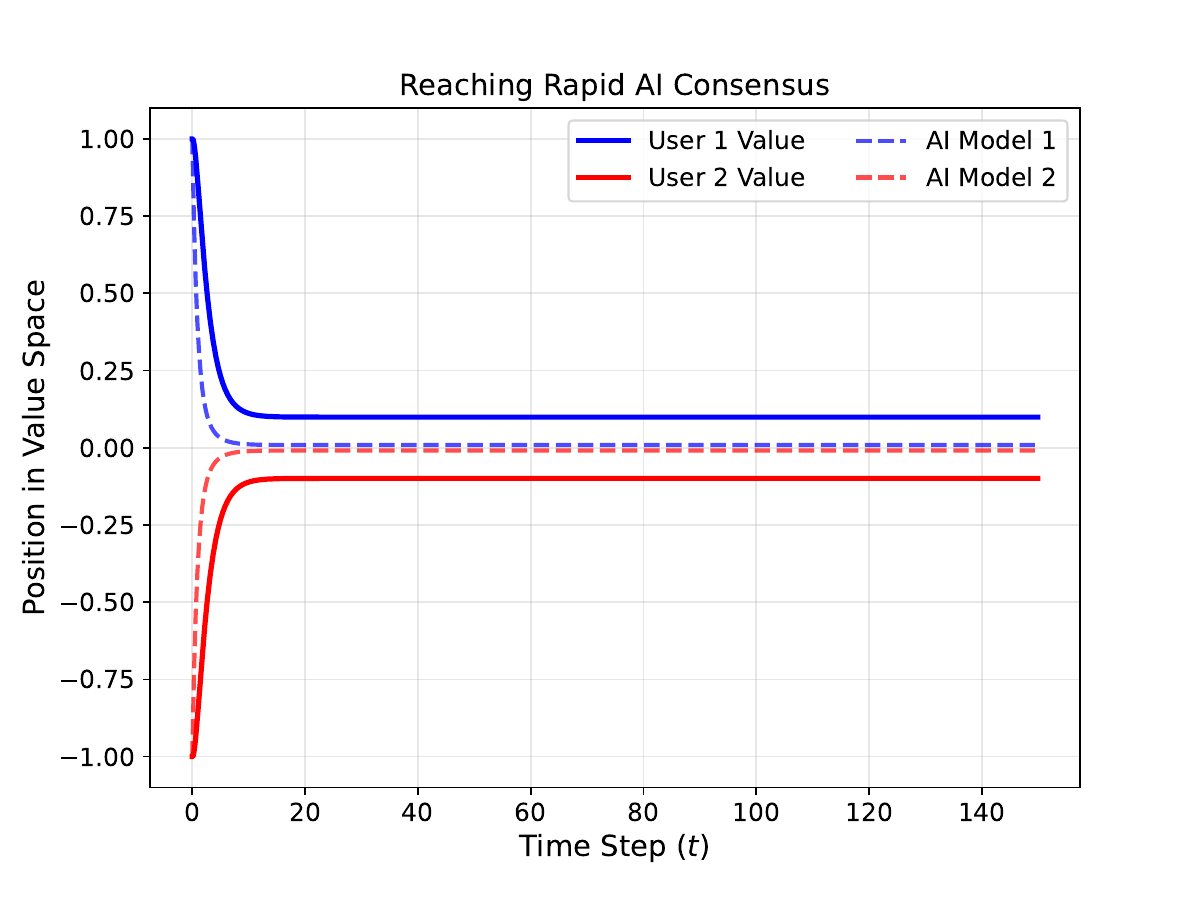}
    \caption{Continuous-time simulation of two completely disconnected human users whose AI assistants are socially coupled via synthetic data loops ($w_{syn} > 0$). The y-axis represents the projection of the high-dimensional value space onto the primary axis separating the initial value vectors of the two users.}
    \label{fig:aiconsensus}
\end{figure}

If AI models exchange synthetic data at rates exceeding human communication, the synthetic coupling weight may potentially exceed the AI learning rate $\lambda$. This is visualised in Figure \ref{fig:aiconsensus}, where we simulate two distinct human users who share no social connections. Their respective AI value models update based on their user's values and their own internal value models through mutual interaction. When $w_{syn} \gg \lambda$, the AI-to-AI consensus term dominates the update equation, and the AI models converge toward each other faster than they track their respective users, resulting in a persistent steady-state gap between each AI's internal model and its assigned user's values.

\subsection{Extending the Number of Normative Environments}

In the baseline model, we are tracking one global and one local environment vector, and each agent takes part in a single localised environment. This kind of abstraction may be reasonable if it is meant to reflect physically segregated populations, but it otherwise fails to generalise to the reality of overlapping identities and online communities. Most individuals in modern societies need to effectively navigate distinct normative expectations of multiple sub-cultures and communities simultaneously \citep{leibo2024theoryappropriatenessapplicationsgenerative,leibo2025societaltechnologicalprogresssewing}.

It is consequently possible to extend our base framework, by replacing the single group assignment with a membership vector $\mathbf{c}_i \in [0, 1]^K$, where $c_{ik}$ denotes the degree to which agent $i$ is embedded in the specific sub-culture $k$.

In this extended model, the agent experiences normative friction from multiple, potentially conflicting environmental pulls. The individualised environmental learning term in the system dynamics expands into a weighted sum as follows:
        
\begin{equation}
\dot{\mathbf{V}}_{i, \text{learning}} = w_{\text{learning}} \sum_{k=1}^K c_{ik} (\mathbf{E}_k - \mathbf{V}_i)
\end{equation}

Correspondingly, the composite utility function is modified such that the agent's localised utility penalty is the weighted sum of distances to each valid group norm:

\begin{equation}
\begin{aligned} U_i = \exp \bigg( - \beta \bigg( &||\mathbf{V}_{i, \text{global}} - \mathbf{E}_{\text{global}}||^2 \\ &+ \sum_{k=1}^K \tilde{c}_{ik} ||\mathbf{V}_{i, \text{local}} - \mathbf{E}_k||^2 \bigg) \bigg) \end{aligned}
\end{equation}

where $\tilde{c}_{ik}$ are the normalised membership weights. This extension allows for future investigations into whether such overlapping group memberships help stabilise the global network against polarisation, or whether they make it harder for individual agents to reach their full target utility.

\section{Extended Simulation Sweeps}

\subsection{The Importance of Continual Exploration}
\label{sec:exp8}

We map the interaction between intrinsic exploration $w_{explore}$ and $\alpha$ in Figure \ref{fig:exp8}. High $\alpha$ requires commensurately high exploration to allow the users to keep up with the drift. Conversely, low $\alpha$ combined with high exploration results in excessive noise that leads to low utility. Interpreting intrinsic exploration as a proxy for normative experimentation, this illustrates how AI alignment strength may act as a dampener on future human innovation. As alignment strength increases, the energy we put towards innovation should increase accordingly.

\begin{figure}[htbp]
    \centering
    \includegraphics[width=\linewidth]{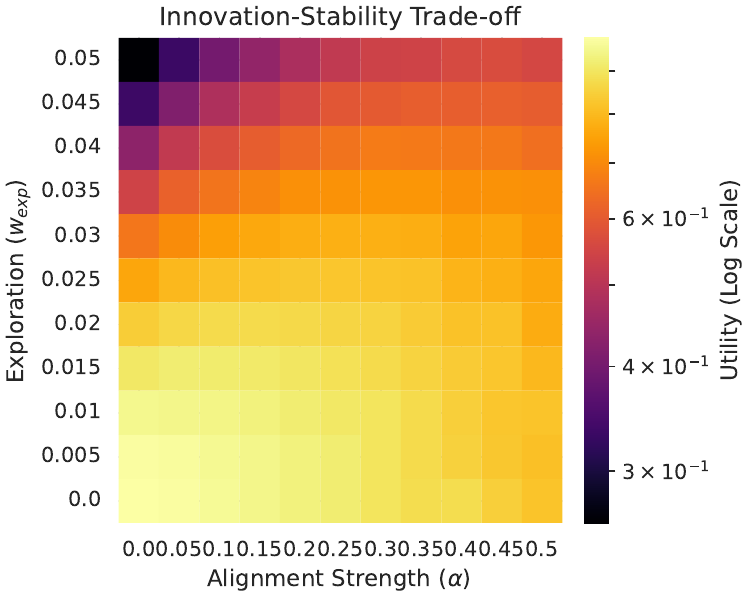}
    \caption{Final utility as a function of intrinsic exploration ($w_{explore}$) and alignment strength ($\alpha$). High alignment requires commensurately high exploration to maintain utility.}
    \label{fig:exp8}
\end{figure}

\subsection{Dynamic Homophily}
\label{sec:homophily}

Here we present further simulation results regarding environmental heterogeneity, in Figure \ref{fig:pca}, where we visualise the population structure in the value space using principal component analysis. In a random network, user-AI pairs from different environmental groups are pulled together into a mixed, homogenised cloud. The social interference from out-group peers prevents distinct sub-cultures from separating, forcing the population toward a global mean - that may be maladaptive for all groups. In contrast, under a dynamically updated kNN network, agents disconnect from maladaptive peers and sort themselves into distinct clusters.

\begin{figure*}[htbp]
    \centering
    \includegraphics[width=\linewidth]{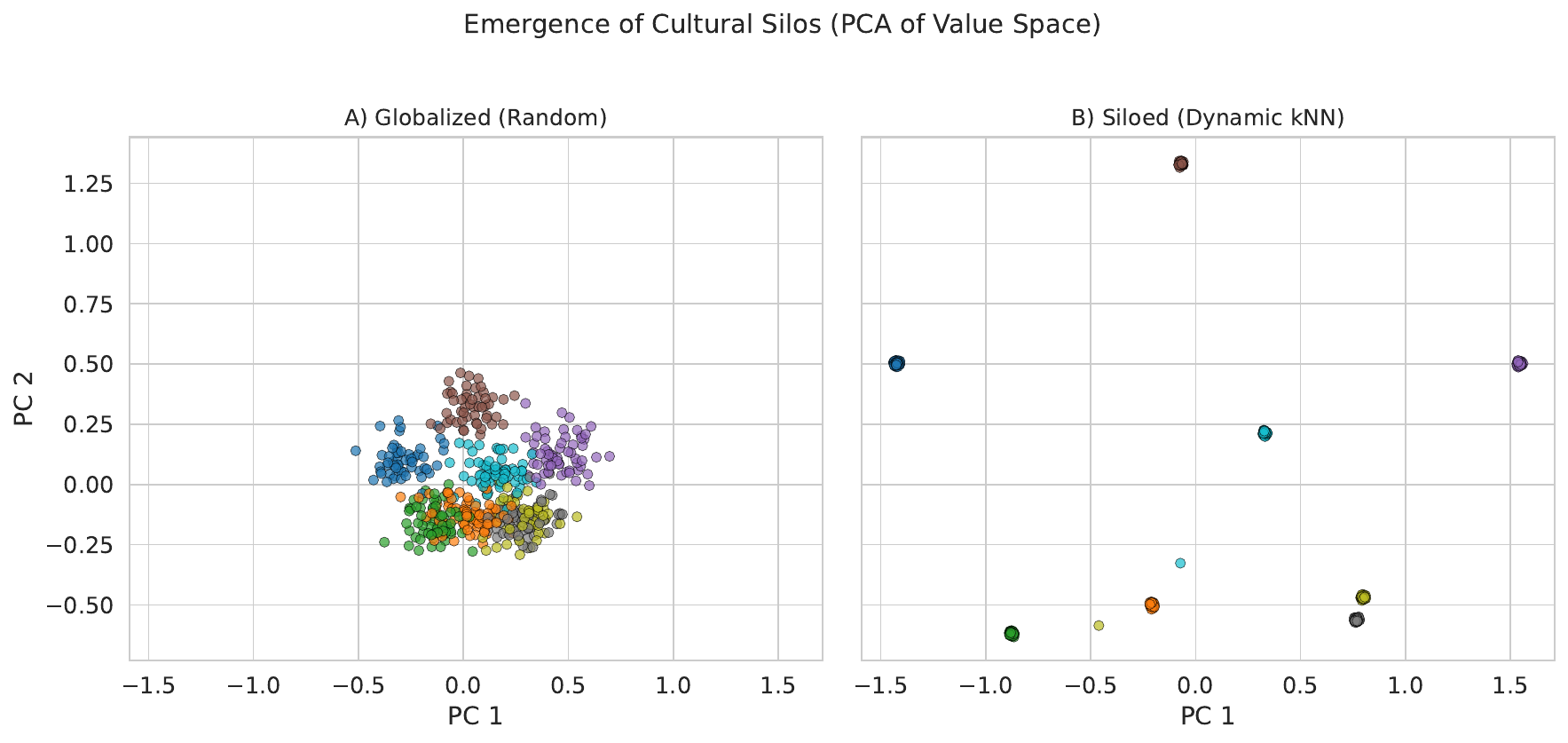}
    \caption{Emergence of cultural silos in kNN networks compared to the random population connectivity. Visualised for the case of $G=10$ distinct groups.}
    \label{fig:pca}
\end{figure*}

In our experiments, we have implemented this dynamic updating by having the agents periodically disconnect from maladaptive peers and reorganise into distinct clusters according to the updated values. This periodic rewiring of the social graph is conducted once for every 10 simulation steps. During this phase, each agent reconnects to a set of $k=10$ nearest neighbours according to the Euclidean distances in the current value space. This continuous topological updating allows agents to dynamically sever ties with peers who have drifted into incompatible normative niches.

This is obviously a fairly simplistic illustration, as real networks operate in ways that are much more complex.

\subsection{Limits of Adaptation}
\label{sec:exp10}

Finally, we note that a simple consequence of our base model is that the learning rate of the human population needs to follow the rate of environmental change. Rather than an insight, this is simply a feature of the model, and one that we would expect to see. For illustration, we sweep the adaptation weight against the environmental drift speed. The results are shown in Figure \ref{fig:exp10}.

\begin{figure}[htbp]
    \centering
    \includegraphics[width=\linewidth]{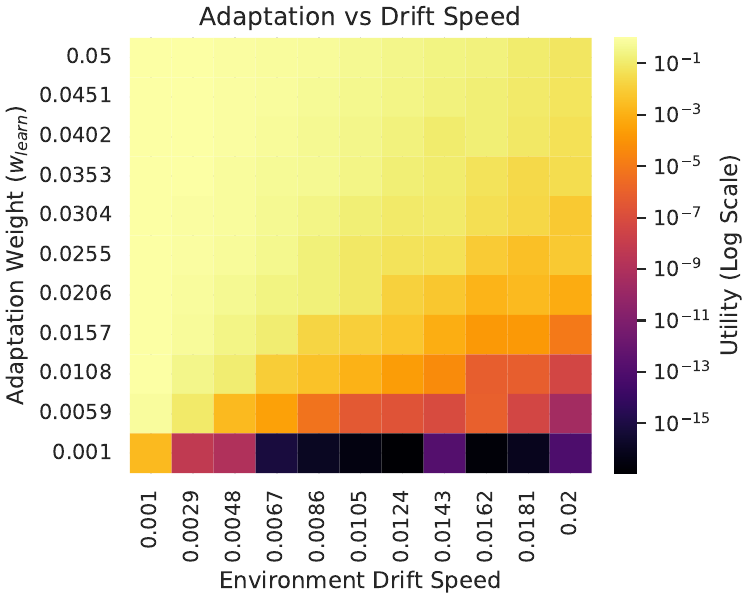}
    \caption{Final utility as a function of learning weight ($w_{learn}$) and environmental drift speed ($\delta$). The boundary delineates where environmental drift outpaces adaptation capacity.}
    \label{fig:exp10}
\end{figure}

\end{document}